\PassOptionsToPackage{table}{xcolor}
\documentclass[12pt]{article}
\usepackage{color}

\pdfoutput=1
\usepackage{hyperref}
\hypersetup{
    colorlinks=true,
    linkcolor=black,
    citecolor=black,
    filecolor=black,
    urlcolor=black,
}

\usepackage{amssymb, amsmath}
\usepackage{tikz-cd}
\usepackage{pgfplots}
\usepackage{calc}

\RequirePackage{amsthm}

\theoremstyle{definition}

\newtheorem{proposition}{Proposition}
\newtheorem{lemma}{Lemma}

\newtheorem{definition}{Definition}
\newtheorem{theorem}{Theorem}
\newtheorem{corollary}{Corollary}

\newcounter{homogSection}
\newcommand{\aAssump}{A\arabic{homogSection}}
\newcounter{het1section}
\newcommand{\bAssump}{\arabic{het1section}}
\newcounter{het2section}
\newcommand{\cAssump}{C\arabic{het2section}}

\newtheoremstyle{theoremSuppressedNumber}{}{}{}{}{\bfseries}{.}{ }{\thmname{#1}\thmnote{ (\mdseries #3)}}
\theoremstyle{theoremSuppressedNumber}

\newtheorem{het1Assump}{Assumption \bAssump \addtocounter{het1section}{1}}

\usepackage{mathtools}
\usepackage{setspace}
\usepackage[longnamesfirst]{natbib}
\usepackage{ulem}
\usepackage{subfigure}
\usepackage{wrapfig}
\usepackage{float}
\usepackage{afterpage}
\newcommand{\E}{\mathbb{E}}
\renewcommand{\P}{\mathbb{P}}
\newcommand{\Prob}{\ensuremath{\mathbb{P}}}
\newcommand{\R}{\ensuremath{\mathbb{R}}}
\usepackage{bbm}
\newcommand{\indicator}{\ensuremath{\mathbbm{1}}}
\makeatletter
\newcommand*{\indep}{
  \mathbin{
    \mathpalette{\@indep}{}
  }
}
\newcommand*{\@indep}[2]{
  \sbox0{$#1\perp\m@th$}
  \sbox2{$#1=$}
  \sbox4{$#1\vcenter{}$}
  \rlap{\copy0}
  \dimen@=\dimexpr\ht2-\ht4-.2pt\relax
  \kern\dimen@
  {#2}
  \kern\dimen@
  \copy0                       
} 
\DeclareMathOperator*{\supp}{supp}
\newcommand{\Exp}{\ensuremath{\mathbb{E}}}
\DeclareMathOperator*{\minimize}{minimize}

\usepackage{epstopdf}
\usepackage{chngpage}
\usepackage{multirow}
\usepackage{tabulary}
\usepackage{tabularx}
\usepackage{booktabs}
\usepackage{textcomp}

\title{\textbf{Assessing Sensitivity to IV Exclusion and Exogeneity without First Stage Monotonicity}\footnote{This paper supersedes Section 5 of the now inactive working paper \cite{MastenPoirier2020}. We thank audiences at the 2024 Southern Economic Association conference, 2025 Winter Meeting of the Econometric Society, and 2025 Greater NY Econometrics Colloquium for helpful conversations and comments. Masten thanks the National Science Foundation for research support under Grant 1943138.}}

\author{
Paul Diegert\thanks{Department of Economics, Toulouse School of Economics, \texttt{paul.diegert@tse-fr.eu}}
\and
Matthew A. Masten\thanks{Department of Economics, Duke University, \texttt{matt.masten@duke.edu}}
\and
Alexandre Poirier\thanks{Department of Economics, Georgetown University, \texttt{alexandre.poirier@georgetown.edu}}
}

\usepackage{titlesec}

\usepackage{titling}
\definecolor{lightgray}{gray}{0.95}
\newcolumntype{a}{>{\columncolor{lightgray}}c}

\begin{document}
\maketitle

\vspace{-2em}

\begin{abstract}
Exclusion and exogeneity are core assumptions in instrumental variable (IV) analyses, but their empirical validity is often debated. This paper develops new sensitivity analyses for these assumptions. Our results accommodate arbitrary heterogeneity in treatment effects and do not impose any monotonicity requirements on the first stage. Specifically, we derive identified sets for the marginal distributions of potential outcomes and their functionals, like average treatment effects, under a broad class of nonparametric relaxations of the exclusion and exogeneity assumptions. These identified sets are characterized as solutions to linear programs and have desirable theoretical properties. We explain how to estimate these solutions using computationally tractable methods even when the linear program is infinite-dimensional. We illustrate these methods with an empirical application to peer effects in movie viewership, using weather as a potentially imperfect instrument.
\end{abstract}

\bigskip
\small
\noindent \textbf{JEL classification:}
C14, C18, C21, C26, C51

\bigskip
\noindent \textbf{Keywords:}
Instrumental Variables, Sensitivity Analysis, Nonparametric Identification, Partial Identification

\allowdisplaybreaks

\onehalfspacing

\newpage

\section{Introduction}\label{sec:intro}

Instrumental variable (IV) analyses typically rely on two core assumptions: Instrument \textit{exclusion} and instrument \textit{exogeneity}. Exclusion holds when the instrument has no direct effect on the outcome, while exogeneity holds when the instrument is randomly assigned. Since the work of \cite{ImbensAngrist1994}, a third assumption is also often imposed: First stage \emph{monotonicity}. In the simplest setting where the treatment and instrument are binary, monotonicity holds if the instrument's effect on treatment is always of the same sign.

All three assumptions can be hard to justify in some empirical settings. Instruments may have direct effects on outcomes or may not be randomly assigned. Monotonicity can also fail. This occurs, for example, in leniency designs (also called `judge IV' designs) where monotonicity implies that a judge is stricter or more lenient in the face of any possible case; see \cite{FrandsenLefgrenLeslie2023Judging} for details. In designs with many treatment and instrument values, there is no single monotonicity assumption to choose from, and it may be difficult to find one suitable for one's empirical setting.

To address these concerns, we study identification of treatment effects in a setting where no monotonicity conditions are imposed whatsoever, but where exogeneity and exclusion are assumed to at least partially hold in some sense. Specifically, we introduce a unifying class of continuous relaxations of instrument exclusion and exogeneity that nests several prominent approaches in the literature. In particular, it includes as special cases the marginal sensitivity model (MSM) of \cite{Tan2006}, $c$-dependence by \cite{MastenPoirier2018}, and supremum distance approaches, as in \cite{Manski1983} and \cite{KlineSantos2013}. All these approaches were developed as sensitivity models for unconfoundedness (or selection on observables), but we develop modified versions suitable for IV sensitivity analysis. In each of those cases, the sensitivity model is indexed by a scalar, unit-free sensitivity parameter that is easy to interpret.

When the outcome variable is discrete, we show that the identified sets for the conditional probabilities of the potential outcomes given the instruments can be characterized as the intersection of two convex sets, parameterized by the relaxation of instrument exclusion and exogeneity.   Using this result, we then show that the identified set for a class of linear functionals of the densities of potential outcomes is the solution to a linear program that can be computed efficiently. This class of functionals includes the standard treatment parameters such as the ATE, the average effect of treatment on the treated (ATT), and quantile treatment effects (QTE). 

We show that these identified sets exhibit many desirable properties, including continuity and monotonicity with respect to the sensitivity parameter. As is well known (e.g., \citealt{BalkePearl1997}), IV models have testable implications, which can fail in practice. In this case, we characterize the smallest deviations from the baseline model that prevent the model from being refuted.

We then extend our results to the case where outcomes are continuous. This case is more delicate, as the distribution of outcomes is now characterized by a density function, an infinite-dimensional object. We show that the identified set for densities and its functionals (like ATE, ATT, or QTE) can also be characterized via a linear program, albeit an infinite-dimensional one. As in the discrete case, we show that identified sets derived from this linear program have desirable properties by analyzing them as infinite-dimensional-valued correspondences, since each sensitivity parameter is now associated with a set of infinite-dimensional-valued density function. As these linear programs cannot be solved in practice, we propose a tractable approach for approximating the problem with a finite-dimensional one.

Using these computational results, we show how applied researchers can produce sensitivity plots that show the sensitivity (or robustness) of their parameter of interest to exclusion and exogeneity violations. These plots can be used, for example, to determine how strong exclusion or exogeneity violations can be before the data is consistent with a zero treatment effect.

To illustrate our approach, we revisit Gilchrist and Sands' \citeyearpar{GilchristSands2016} study of peer effects in movie viewership, using weather as an instrument for opening-weekend viewership. While extremely popular in empirical practice, weather instruments have come under increasing scrutiny in recent years (e.g., \citealt{Sarsons2015}, \citealt{GallenRaymond2023}, \citealt{Mellon2025}). In this application, social learning and dynamic behavior could lead to violations of instrument exclusion, and use our results to assess the robustness of conclusions to relaxations of this assumption. Using both discretized and continuous outcomes, we confirm that under the baseline of instrument exclusion, there is a positive peer effect on viewership, but show that this conclusion is sensitive to relatively small relaxations of the exogeneity assumption.

The rest of the paper is organized as follows. We first provide an overview of the related literature. Section \ref{sec:binary_outcome} then develops the framework for binary outcomes, introduces the relaxation class, and derives sharp identified sets, falsification frontiers, and falsification adaptive sets in the discrete setting. Section \ref{sec:cont_Y} extends the analysis to continuous outcomes, establishes the corresponding identification and continuity results, and details a sieve-based computational strategy. Section \ref{subsec:empirical} presents our empirical application.

\subsubsection*{Related Literature}\label{subsec:litreview}

Research on the sensitivity of IV results to violations of exclusion and exogeneity go back to at least \cite{Fisher1961}. More recent developments were proposed by \cite{BoundJaegerBaker1995}, \cite{Small2007} and  \cite{ConleyHansenRossi2012}. All of these methods assume a linear outcome equation, motivated by treatment effect homogeneity, which we do not assume. These papers bound the direct effect of the instrument on the outcome, which can be done by bounding the coefficient on the instrument under the assumption that the potential outcomes depend linearly on it. Various approaches have been proposed to bound this direct effect, including \cite{NunnWantchekon2011}, \cite{ConleyHansenRossi2012}, \cite{Kraay2012},  \cite{KippersluisRietveld2017}, \cite{KippersluisRietveld2018}, and \cite{MastenPoirier2021}.  Also see \cite{AltonjiElderTaber2005a},  \cite{Ashley2009}, and \cite{AshleyParmeter2015} for alternative approaches.

Our paper contributes to the literature on sensitivity analysis in instrumental variable models with heterogeneous treatment effects, which is much sparser than that for homogeneous treatment effects. Specifically, few papers consider continuous relaxations of the baseline instrumental variable assumptions while still allowing for heterogeneous treatment effects.

Early work by \citet{Manski1990} characterizes sharp bounds on average treatment effects under two sets of assumptions: (i) instrument exclusion and exogeneity hold (formulated as mean independence in his general analysis) or (ii) instrument exclusion and exogeneity fail arbitrarily. Our continuously parameterized sensitivity model spans these two sets of assumptions, allowing users to calibrate the degree of exclusion and exogeneity violations from ``no-violations'' (i.e., full exclusion and exogeneity) to ``no assumptions''.

\cite{HotzMullinSanders1997} used a mixture model to allow for relaxations of the baseline assumptions. They focus on the average effect of treatment on the treated, whereas our sensitivity analysis allows for a broader set of parameters of interest. \cite{Ramsahai2012} studied a heterogeneous treatment effect model with a binary outcome, binary treatment, and a binary instrument. He defines a continuous relaxation of the instrument exogeneity assumption and then shows how to numerically compute identified sets for a single value of this relaxation. On pages 842-–843, he notes that ``it is not obvious how the methods described in [his] paper can be extended to compute bounds'' as a function of his relaxation. In our analysis, we allow all variables to be nonbinary, and even continuous for the outcome variable, and we allow for multiple instruments. We also consider a large set of target parameters and derive theoretical and computational properties for the sensitivity plots, which map the sensitivity parameters into this range of target parameters. Also see \cite{Huber2014} and \cite{MachadoShaikhVytlacil2019} for related, but different, approaches.

In fully discrete cases, identified sets for causal parameters and counterfactual distributions can often be obtained via linear programming. This observation goes back at least to \citet{BalkePearl1997} (and related work by \citealt{Pearl1995}) and is emphasized in more recent reviews of discrete partial identification methods. For example, see the literature review in \cite{Torgovitsky2019}. Linear programming has been used in several papers to do sensitivity analysis. One paper is \cite{Ramsahai2012}, which we already discussed above. \citet[section 4]{Laffers2019} considers continuous relaxations of instrument exogeneity. He then computes identified sets for ATE for several values of this relaxation. In \cite{Laffers2018}, he applies this approach to various additional forms of continuous relaxations. \cite{Duarte2024} also uses linear programming to bound parameters under exclusion and monotonicity violations. These papers all require all variables to be discrete. A key contribution of our paper is that our results allow for continuous outcome variables. 

Our paper also contributes to assessments of IV model falsification. \cite{BalkePearl1997} characterize when Manski's bounds are empty, and hence when the model is falsified.\footnote{\cite{BalkePearl1997} assume the instrument is independent of the potential outcomes jointly, whereas \cite{Manski1990} only assumed the instrument is independent of each potential outcome separately. (Here we suppose outcomes are binary, so that mean independence is equivalent to statistical independence.) This difference does not affect whether the identified set is empty, given any fixed distribution of observables. Hence, it does not change the testable implications of the model. When the identified set is nonempty, however, this difference \emph{can} affect its size. See the second paragraph of section 3 in \cite{SwansonHernanMillerRobinsRichardson2018} for further discussion.} \citet[Proposition 3.1]{Kitagawa2021} generalizes this characterization to allow for continuous outcomes, still requiring the treatment and instrument to be binary. As \cite{Kitagawa2021} notes, his extension is an adaptation of Corollary 2.2.1 in Manski's \citeyearpar{Manski2003} analysis of missing data. \citet[Proposition 2.4]{BeresteanuMolchanovMolinari2012} further generalizes this characterization to allow for continuous instruments and discrete treatment, for discrete or continuous outcomes. \citet[Proposition 1]{KedagniMourifie2020} provides an alternative characterization when instruments and outcomes are continuous, treatment is binary, and under the stronger assumption that the instrument is independent of the potential outcomes jointly; also see Proposition 2.5 of \cite{BeresteanuMolchanovMolinari2012} for a result under this stronger independence assumption.

Finally, a large literature on the testable implications of instrument exclusion and exogeneity, combined with other assumptions, has developed. Most notably, many papers have studied the testable implications of the monotonicity assumption of \cite{ImbensAngrist1994}. \cite{FloresChen2018} give a comprehensive review. Also see \cite{FrandsenLefgrenLeslie2023Judging} for discussions of monotonicity in the judge IV framework. In this paper, we focus on instrument exclusion and exogeneity only.

\section{Sensitivity Analysis with Binary Outcomes}\label{sec:binary_outcome}

We begin by considering analyses with a binary outcome. For further simplicity, we also assume that the treatment and instrument are binary. The results below generalize to a setting with multiple treatment values and multiple discrete instruments, but we focus on the binary case, which allows us to explain the main ideas and results while keeping the notation simple. See Section \ref{subsec:discrete_generalization} for their generalization. The case where the outcome variable is continuously distributed presents additional technical challenges and is analyzed in Section \ref{sec:cont_Y}.

\subsection{Model, Parameters of Interest, and Assumptions}
Let $X \in \{ 0, 1 \}$ denote the observed binary treatment variable and $Z \in \{0,1 \}$ denote an observed instrument. As mentioned above, we consider multiple treatments and discrete instruments later. Let $\{Y(x,z)\}_{x,z \in \{0,1\}}$ denote potential outcomes for both treatment and instrument values. The observed outcome is denoted by
\begin{equation}\label{eq:generalOutcomeEquation}
 Y = Y(X,Z).
\end{equation}
We assume the joint distribution of $(Y,X,Z)$ is known in this identification analysis. Our analysis could be done conditional on a vector $W$ of covariates, but we omit them for simplicity.  Let $p_Z \coloneqq \Prob(Z=1)$ and $\pi(x \mid z) \coloneqq \Prob(X = x \mid Z = z)$. We maintain the following assumption to rule out trivial cases.
\begin{het1Assump}\label{assn:NonTrivialinstrument}
 Let $p_Z \in (0,1)$ and $\pi(x \mid z) \in (0,1)$ for all $x,z \in \{0,1\}$.
\end{het1Assump}

We define $p_Y(x,z) \coloneqq \Prob(Y(x,z) = 1\mid Z = z)$, conditional probabilities of the potential outcomes given the instrument. Let $\textbf{p}_Y(x) \coloneqq (p_Y(x,0),p_Y(x,1)) \in [0,1]^2$ and $\textbf{p}_{Y} \coloneqq (\textbf{p}_Y(0),\textbf{p}_Y(1)) \in [0,1]^4$ be collections of these conditional probabilities. We are interested in functionals of these conditional probabilities, denoted by $\Gamma: [0,1]^4 \rightarrow \R$, which include various treatment effect parameters. In this section, we focus our attention on averages of treatment effects such as the average treatment effect (ATE) and the average treatment effect on the treated (ATT). They can be viewed as functionals of $\textbf{p}_Y$ as follows:
\begin{align*}
 \text{ATE} &\coloneqq \E[Y(1,Z) - Y(0,Z)] = \Gamma_\text{ATE}(\textbf{p}_{Y})\\
 \text{ATT} &\coloneqq \E[Y(1,Z) - Y(0,Z) \mid X = 1] = \Gamma_\text{ATT}(\textbf{p}_{Y}),
\end{align*}
where 
\begin{align}
 \Gamma_{\text{ATE}}(\textbf{p}_{Y}) &\coloneqq  p_Y(1,1)p_Z + p_Y(1,0)(1-p_Z) - p_Y(0,1)p_Z - p_Y(0,0)(1-p_Z)\label{eq:ATE_functional}\\
 \Gamma_{\text{ATT}}(\textbf{p}_{Y}) &\coloneqq \E[Y \mid X = 1] - \frac{p_Y(0,1)p_Z + p_Y(0,0)(1-p_Z) - \E[Y \mid X = 0](1-p_Z)}{p_Z}.\label{eq:ATT_functional}
\end{align}
These parameters are well-defined even in the absence of exclusion or exogeneity assumptions about the instruments. Additional parameters could be of interest, such as the local average treatment effect (LATE). The LATE is defined in terms of potential treatments, which could be incorporated into our framework but are not required by it. 

Before introducing additional assumptions, we first characterize the identified set  for $\textbf{p}_{Y}$ when no assumptions are made about the joint distribution of $(Y,X,Z)$ beyond the regularity assumption \ref{assn:NonTrivialinstrument}. To do so, define
\begin{align}
\label{eq:boxSet}
 \mathcal{H}_x
 &\coloneqq [\P(Y=1,X=x \mid Z=0), \ \P(Y=1,X=x \mid Z=0) + \pi(1-x \mid 0)]\notag \\
 &\qquad \times [\P(Y=1,X=x \mid Z=1), \ \P(Y=1,X=x \mid Z=1) + \pi(1-x \mid 1)], 
\end{align}
which depends on the joint distribution of $(Y,X,Z)$. With this notation, we obtain the following result which is adapted from \cite{Manski1990}.
\begin{proposition}[\citealt{Manski1990}]\label{prop:noassn_discrete}
 Suppose Assumption \ref{assn:NonTrivialinstrument} holds. Then the identified set for $\textbf{p}_Y$ is $\mathcal{H}_0 \times \mathcal{H}_1$.
\end{proposition}
This result shows that the identified set for the conditional probabilities $\textbf{p}_Y$ is a Cartesian product of intervals, i.e., a hyperrectangle. As these bounds are sharp, they can be used to obtain sharp bounds on any functional of $\textbf{p}_Y$. For example, the functional $\Gamma_{\text{ATE}}$ is linear and the set $\mathcal{H}_0 \times \mathcal{H}_1$ is a Cartesian product of intervals, so appropriately evaluating $\Gamma_{\text{ATE}}$ at the lower/upper bounds of the intervals in \eqref{eq:boxSet} will yield sharp bounds for it. The same approach can be used to obtain sharp bounds of the ATT, for example. For any linear functional, this is equivalent to a linear program, which is easy to solve analytically given the discrete supports of $X$ and $Z$. Figure \ref{fig:FFbinaryYsingleIV} illustrates this identified set and the optimization of the ATE over the identified set for $\textbf{p}_Y(x)$.

\begin{figure}[t]
\centering
\includegraphics[width=75mm]{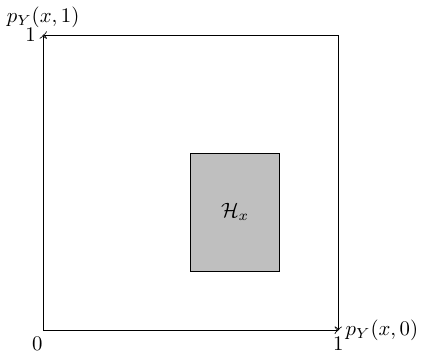}
\hspace{5mm}
\includegraphics[width=75mm]{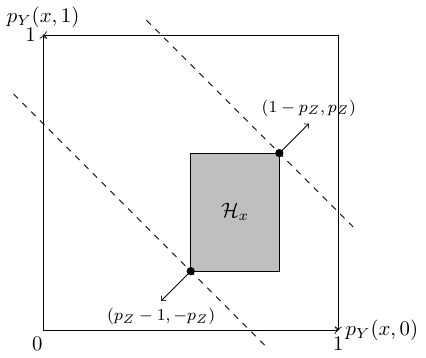}
\footnotesize
\caption{Left: Example identified set for $\textbf{p}_Y(x) = (p_Y(x,0),p_Y(x,1))$ under no exogeneity assumptions. Right: corresponding linear program minimizing/maximizing the ATE $= p_Y(x,0)(1-p_Z) + p_Y(x,1)p_Z$.}
\label{fig:FFbinaryYsingleIV}
\end{figure}

These bounds can be considerably tightened by assuming exogeneity or exclusion, as we formally define below.

\subsubsection*{Baseline Assumptions}

We now introduce the assumptions we will study in this model. We compare these assumptions to the four assumptions usually imposed in a large segment of the literature, including the traditional Local Average Treatment Effect (LATE) framework: exogeneity, exclusion, monotonicity, and relevance. For brevity, we do not include covariates in this discussion, although all the upcoming assumptions can be stated conditional on a covariate vector $W$. 

First, we formally define the exogeneity and exclusion assumptions we consider. 
\begin{definition}[Exogeneity] \label{def:exogeneity}
 The instrument is exogenous if $Z \indep Y(x,z)$ holds for each $(x,z) \in \supp(X) \times \supp(Z)$.
\end{definition}

Exogeneity holds when the instrument is randomly assigned, or as good as randomly assigned, with respect to the potential outcomes. We do not require that the instrument be independent of potential treatment values, although this assumption can be incorporated into the framework. As mentioned earlier, we could consider relaxing the conditional exogeneity assumption $Y(x) \indep Z \mid W$, at the cost of additional notation.

Next, we consider an exclusion assumption that is weaker than the most commonly used version.
\begin{definition}[Weak Exclusion]\label{def:exclusion}
 The instrument is weakly excluded if $Y(x,z) \mid \{Z = z''\} \overset{d}{=} Y(x,z') \mid \{Z = z''\}$ for all $x \in \supp(X)$ and $z,z',z'' \in \supp(Z)$.
\end{definition}

The standard exclusion assumption is that $Y(x,z) = Y(x,z')$ with probability 1 for any possible treatment value $x$ and any possible instrument values $z$ and $z'$, whereas weak exclusion only requires the (conditional) distributions of these potential outcomes to be identical. This has also been called \textit{stochastic exclusion}; see, for example, \cite{SwansonHernanMillerRobinsRichardson2018}.  Although we do not study the LATE here, the arguments used to obtain a causal interpretation for the Wald estimand are not impacted if exclusion is replaced by weak exclusion. In particular, the Wald estimand equals the LATE when the treatment and instrument are binary under weak exclusion, provided appropriate exogeneity, relevance, and monotonicity conditions hold.

We will assume that the instrument is exogenous or weakly excluded, without requiring it to satisfy both.
\begin{het1Assump}\label{assn:exog_or_excl}
 The instrument $Z$ is exogenous or weakly excluded.
\end{het1Assump}
Under this assumption, $p_Y(x,z)$ can be interpreted in one of two ways. Under exogeneity, this probability equals the unconditional probability $\P(Y(x,z) = 1)$, while under weak exclusion it denotes the conditional probability $\P(Y(x) = 1 \mid Z = z)$. If both hold, then $p_Y(x,z)$ does not depend on $z$, meaning that $\P(Y(x,1) = 1 \mid Z = 1) = \P(Y(x,0) = 1 \mid Z = 0)$, but Assumption \ref{assn:exog_or_excl} allows the dependence of $p_Y(x,z)$ on $z$ to be nontrivial. We formally show that under Assumption \ref{assn:exog_or_excl}, $p_Y(x,z)$ not depending on $z$ implies the exogeneity \textbf{and} weak exclusion of the instrument.

\begin{lemma}[Condition for Exogeneity and Weak Exclusion]\label{lemma:equivalence}
 Suppose assumptions \ref{assn:NonTrivialinstrument} and \ref{assn:exog_or_excl} hold. Then $p_Y(x,1) = p_Y(x,0)$ for all $x \in \{0,1\}$ if and only if $Z$ is exogenous and weakly excluded.
\end{lemma}

Thus, we can view failures of exogeneity or exclusion as mathematically equivalent to the probabilities $p_Y(x,z)$ being nonconstant in $z$.  

To simplify our exposition going forward, we let
\[
	Y(x) \coloneqq Y(x,Z).
\]
Note that $p_Y(x,z) = \P(Y(x) = 1\mid Z=z)$, and that the ATE and ATT functionals $\Gamma_{\text{ATE}}$ and $\Gamma_{\text{ATT}}$ are defined as functionals of $Y(x)$, independently of whether weak exclusion or exogeneity holds.  Also, the bounds of Proposition \ref{prop:noassn_discrete} do not change when Assumption \ref{assn:exog_or_excl} is imposed.  With these definitions, we can see that the instrument is exogenous and weakly excluded if and only if 
\[
	Y(x) \indep Z
\]
for $x \in\{0,1\}$. Hence, we will consider relaxations of exogeneity or weak exclusion as relaxations of an independence assumption, as they are mathematically equivalent here.

To finish our comparison to the standard IV assumption, we note that we allow for positive masses of both compliers, i.e., units for whom $X(1) > X(0)$, and defiers, i.e., units for whom $X(0)>X(1)$. Again, this assumption could be added to our framework at the cost of additional notation, but we focus on the case where no restrictions are imposed on these potential treatments. We also do not require that $\Prob(X=1 \mid Z=1) \neq \Prob(X=1\mid Z=0)$, the usual relevance assumption assumed in the LATE framework.

We will maintain assumptions \ref{assn:NonTrivialinstrument} and \ref{assn:exog_or_excl}, and consider a continuum of exogeneity and exclusion assumptions that range from no assumptions to full exogeneity and exclusion. Thus, they will include the no-assumptions case and the case where $Z$ is excluded and exogenous as special cases.

\subsection{Sensitivity Models for the Exogeneity or Exclusion Assumptions}\label{subsec:partial_indep_assn}

We now consider a menu of assumptions that can be interpreted as relaxations of the exogeneity or exclusion assumption. The results of Proposition \ref{prop:noassn_discrete} show one extreme: bounds under no dependence assumptions. We briefly consider bounds under the other extreme, where exogeneity and weak exclusion exactly hold.  \cite{Manski1990} derived the identified set for $\P(Y(x) = 1)$ for $x \in \{0,1\}$ as well as the identified set for the ATE under this assumption.\footnote{Manski's \citeyearpar{Manski1990} analysis considered a general case which does not require outcomes, treatment, or instruments to be binary. In this general setting, he used a mean independence assumption. When outcomes are binary, mean independence of $Y(x)$ from $Z$ is equivalent to statistical independence of $Y(x)$ and $Z$.} 

Under this assumption, $p_Y(x,0) = p_Y(x,1)$ for $x \in \{0,1\}$ by Lemma \ref{lemma:equivalence}. This restricts the probabilities $\textbf{p}_Y$ to lie in the set
\begin{align*}
 \mathcal{A}_\text{indep} &\coloneqq \{(p_{00},p_{01},p_{10},p_{11}) \in [0,1]^4 : p_{00} = p_{01}, p_{10} = p_{11}\}.
\end{align*}
Therefore, the identified set for $\textbf{p}_Y$ is given by the set of probabilities as restricted by the observed distribution of $(Y,X,Z)$, namely $\mathcal{H}_0 \times \mathcal{H}_1$, intersected with the set of probabilities satisfying exclusion and exogeneity, given by $\mathcal{A}_\text{indep}$. Thus, the identified set for $\textbf{p}_Y$ is 
\begin{align}\label{eq:IDset_indep}
 (\mathcal{H}_0 \times \mathcal{H}_1) \cap \mathcal{A}_\text{indep},
\end{align}
which can also be written as
\begin{align*}
 &\bigcap_{z=0,1}[\Prob(Y=1,X=0 \mid Z=z), \ \Prob(Y=1,X=0 \mid Z=z) + \pi(1 \mid z)]\\
 &\qquad \times \bigcap_{z=0,1}[\Prob(Y=1,X=1 \mid Z=z), \ \Prob(Y=1,X=1 \mid Z=z) + \pi(0 \mid z)].
\end{align*}
These bounds take the form of intersections. \cite{Pearl1995} and \cite{BalkePearl1997} showed that this identified set can be empty, and hence that this model is falsifiable. This set is empty if and only if the two sets in \eqref{eq:IDset_indep} are disjoint. An empty identified set corresponds to a falsification of the original model, or of an exogeneity or exclusion assumption, when other model assumptions are maintained. Figure \ref{fig:IDset_indep} shows the identified set both when the model is falsified (left panel), and when it is not (right panel).

\begin{figure}[t]
\centering
\includegraphics[width=77mm]{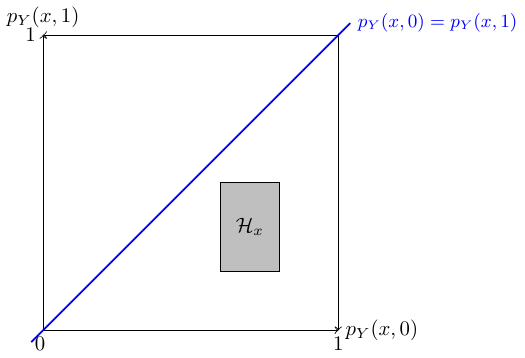}
\hspace{5mm}
\includegraphics[width=77mm]{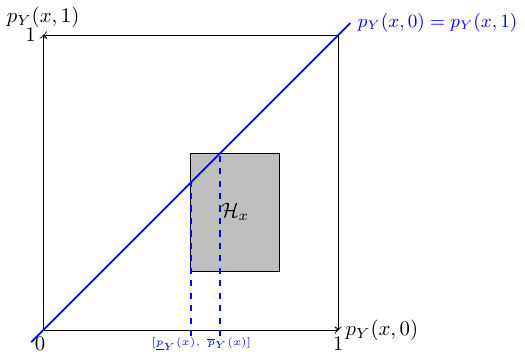}
\footnotesize
\caption{Left: Empty identified set under exogeneity and exclusion. Right: Nonempty identified set under exogeneity and exclusion. The upper and lower bounds for $p_Y(x,1) = p_Y(x,0)$ are denoted by $\overline{p}_Y(x)$ and $\underline{p}_Y(x)$}
\label{fig:IDset_indep}
\end{figure}

Full exogeneity or exclusion of the instrument may be a strong assumption in contexts where we do not believe that $Z$ is assigned randomly, or if we cannot rule out a direct effect of the instrument on the potential outcomes. In these cases, relaxing exogeneity or exclusion is appropriate. The no-assumption bounds of Manski remain valid, but partial validity of the instrument will yield intermediate bounds that are potentially significantly narrower than those in Proposition \ref{prop:noassn_discrete}.

We will consider relaxations of exogeneity and exclusion by characterizing sets of conditional probabilities $\textbf{p}_{Y}$. A large literature on sensitivity analysis has proposed various approaches for relaxing assumptions, often independence or conditional independence assumptions. We will focus on three examples, which are special cases of a unifying class of relaxations from independence we define in Section \ref{subsec:unifying_sensitivity_model_discrete}.

\subsubsection{Marginal Sensitivity Model: Tan (2006)}
The \textit{Marginal Sensitivity Model} (MSM) of \cite{Tan2006} consists of a class of relaxations of an independence assumption between a potential outcome and a binary treatment. It is generalized to multivariate treatments in \cite{ZhaoSmallBhattacharya2019} and \cite{BasitLatifWahed2023}. We consider a version of the MSM that constrains the dependence of the potential outcomes on the instruments, rather than the treatment.

\begin{definition}\label{def:MSModel}
Let $\Lambda \in [1,+\infty]$ be a known sensitivity parameter. The distribution of $(\{Y(x)\}_{x \in \supp(X)},Z)$ satisfies the Marginal Sensitivity Model with parameter $\Lambda$ if
\begin{align}\label{eq:MSModel}
 \frac{\Prob(Z=z)}{\Prob(Z=z')} \Big/ \frac{\Prob(Z=z \mid Y(x) = y)}{\Prob(Z=z' \mid Y(x) = y)} \in \left[\Lambda^{-1},\Lambda\right]
\end{align}
for all  $x \in \supp(X)$, $y \in \supp(Y(x))$, and $z, z' \in \supp(Z)$.\footnote{We let $\Lambda^{-1} = 0$ when $\Lambda = +\infty$.}
\end{definition}

With a binary instrument, this restriction places a bound on the odds ratio between the conditional odds of the instrument $\Prob(Z=z \mid Y(x) = y)/\Prob(Z=1-z \mid Y(x) = y)$ and its unconditional counterpart, $\Prob(Z=z)/\Prob(Z=1-z)$, for $z = 0,1$. In the binary outcome and instrument setting, equation \eqref{eq:MSModel} can be rearranged as
\begin{align}\label{eq:MSM_constraints}
 \frac{\P(Y(x) = y \mid Z = 1)}{\P(Y(x) = y \mid Z = 0)} = \frac{p_Y(x,1)^y(1-p_Y(x,1))^{1-y}}{p_Y(x,0)^y(1-p_Y(x,0))^{1-y}}  \in \left[\Lambda^{-1},\Lambda\right]
\end{align}
for $x,y \in \{0,1\}$. When $\Lambda = 1$, this ratio is 1 and $p_Y(x,1) = p_Y(x,0)$ for $x = 0,1$. By Lemma \ref{lemma:equivalence}, this means that weak exclusion \textit{and} exogeneity hold when $\Lambda = 1$ under Assumption \ref{assn:exog_or_excl}. When $\Lambda = +\infty$, these inequalities do not impose any restrictions on $\textbf{p}_Y$. Intermediate values of $\Lambda$ yield intermediate levels of restrictions on $\textbf{p}_Y$. Note that we can choose different $\Lambda$ values for $x = 0,1$, but we omit this generalization for brevity.

We note that equation \eqref{eq:MSM_constraints} can be written as four linear constraints on $\textbf{p}_Y$ by varying $y$ and $x$ over their support. This will be useful for casting this sensitivity analysis exercise as a linear program, as linear programming is a reliably fast and scalable computation method whose implementation is standard. Define
\begin{align*}
 A_\text{MSM}(\lambda) &\coloneqq \begin{pmatrix}
 1-\lambda & -1\\
 -1 & 1-\lambda\\
 \lambda - 1& 1\\
 1 & \lambda -1
 \end{pmatrix}, \qquad a_\text{MSM}(\lambda) \coloneqq \begin{pmatrix}
 0\\
 0\\
 \lambda\\
 \lambda
 \end{pmatrix}
\end{align*}
where $\lambda \coloneqq 1- \Lambda^{-1}$. The set of conditional probabilities satisfying the marginal sensitivity model with sensitivity parameter $\lambda$ is 
\begin{align}\label{eq:MSM_discrete_correspondence}
 \mathcal{A}_\text{MSM}(\lambda) &\coloneqq \{\textbf{p} \in [0,1]^2: A_\text{MSM}(\lambda)\textbf{p} \leq a_\text{MSM}(\lambda)\}^2
\end{align}
where the weak inequality in \eqref{eq:MSM_discrete_correspondence} is component-wise. We reparametrized the sensitivity parameter $\Lambda$ as $\lambda = 1 - \Lambda^{-1}$ to standardize its scale to $[0,1]$. Here $\Lambda = 1$, or full exogeneity and exclusion, maps into $\lambda = 0$ while $\Lambda = +\infty$, or no assumptions, maps into $\lambda = 1$.

\subsubsection{$c$-dependence}
Introduced in \cite{MastenPoirier2018}, $c$-dependence imposes a bound on the maximum difference between the conditional probability of receiving a binary treatment $\Prob(X=1\mid Y(x))$ and its unconditional probability $\Prob(X=1)$. This was proposed in a setting where the unconfoundedness of treatment $X$ is relaxed. We adapt this sensitivity model to the case where exogeneity or exclusion of an instrument is relaxed. Here is the formal definition of this sensitivity model.
\begin{definition}\label{def:c-dep}
Let $c\in[0,1]$ be a known sensitivity parameter. The distribution of $(\{Y(x)\}_{x \in \supp(X)},Z)$ satisfies \emph{$c$-dependence} if
\begin{equation}\label{eq:c-indep1}
 |\P(Z=z \mid Y(x) = y) - \Prob(Z=z) | \leq c
\end{equation}
for all $y \in \supp(Y(x))$, $x \in \supp(X)$, and $z \in \supp(Z)$.
\end{definition}
When $Z$ is binary, it suffices to impose this inequality for $z=1$ only. When $c=0$, $c$-dependence is equivalent to imposing full exclusion and exogeneity. Values of $c$ exceeding $\max\{ p_Z,1-p_Z \}$ do not constrain the stochastic relationship between $Z$ and $Y(x)$. $c \in (0,1)$ partially constrains the stochastic relationship between $Z$ and $Y(x)$. \cite{MastenPoirier2023} give additional discussion of how to interpret $c$-dependence.

We can again rewrite the above restriction into a system of four linear restrictions on $\textbf{p}_Y$. Let
 \begin{align*}
 A_\text{$c$-dep}(c) &\coloneqq \begin{pmatrix}
  k_0(c) & -1\\
  -1 & k_1(c)\\
  -k_0(c) & 1\\
  1 & -k_1(c)
 \end{pmatrix} \qquad \text{ and } \qquad a_\text{$c$-dep}(c) \coloneqq \begin{pmatrix}
  0 \\
  0 \\
  1- k_0(c)\\
  1 - k_1(c)
 \end{pmatrix},
\end{align*}
where $k_z(c) \coloneqq \frac{\Prob(Z=z)\max\{\Prob(Z=1-z) - c,0\}}{\Prob(Z = 1-z)\min\{\Prob(Z=z)+c,1\}}$ for $z = 0,1$ and $c \in [0,1]$. We can show that the set of conditional probabilities consistent with $c$-dependence with sensitivity parameter $c$ is
\begin{align}\label{eq:cdep_discrete_correspondence}
 \mathcal{A}_\text{$c$-dep}(c) &\coloneqq \{\textbf{p} \in [0,1]^2: \textbf{A}_\text{$c$-dep}(c) \textbf{p} \leq a_\text{$c$-dep}(c)\}^2.
\end{align}
This set depends only on $c$ and $p_Z$. 

\subsubsection{Kolmogorov-Smirnov Distance}
Consider a sensitivity model bounding a metric between the distributions of $Y(x) \mid \{Z=0\}$ and $Y(x) \mid \{Z=1\}$. This type of restriction was used in \cite{KlineSantos2013} to relax a missingness at random assumption. It was also considered for estimation in \cite{Manski1983}. 

\begin{definition}\label{def:KSmodel}
Let $K \in [0,1]$ be a known sensitivity parameter. The distribution of $(\{Y(x)\}_{x \in \supp(X)},Z)$ satisfies the Kolmogorov-Smirnov (KS) model if 
\begin{align}\label{eq:KSdist}
 |\Prob(Y(x) \leq y\mid Z=z) - \Prob(Y(x) \leq y\mid Z=z')| \leq K
\end{align}
for all $x \in \supp(X)$, $y \in \R$, and $z,z' \in \supp(Z)$.
\end{definition}
When outcomes and instruments are binary, this sensitivity model is equivalent to bounding the magnitude of the difference between $\P(Y(x) = 1 \mid  Z = 1)$ and $\P(Y(x) = 1 \mid  Z = 0)$ by $K$. This assumption directly bounds the maximum deviation between the potential outcomes distribution given the instrument's two values. As in the previous two definitions, this class of restrictions encompasses independence ($K = 0$), no assumptions ($K = 1$), and intermediate cases ($K \in (0,1)$). 

The set of conditional probabilities satisfying the Kolmogorov-Smirnov restrictions is characterized by the two linear inequalities
\begin{align}\label{eq:KS_discrete_correspondence}
 \mathcal{A}_\text{KS}(K) \coloneqq \{\textbf{p} \in [0,1]^2 : A_\text{KS} \textbf{p} \leq a_\text{KS}(K)\}^2
\end{align}
where
\begin{align}\label{eq:KS_Matrix}
 A_\text{KS} &\coloneqq \begin{pmatrix}
 1 & -1\\
 -1 & 1\\
 \end{pmatrix} \qquad \text{ and } \qquad a_\text{KS}(K) \coloneqq \begin{pmatrix}
 K\\
 K
 \end{pmatrix}.
\end{align}

\subsection{A Unifying Sensitivity Model}\label{subsec:unifying_sensitivity_model_discrete}

We now consider a general sensitivity model that encompasses the previous three sensitivity models as special cases. We will derive our main theoretical results under this sensitivity model. We assume that $X$ and $Z$ are binary for ease of notation and discuss the generalization to discrete $X$ and $Z$ in Section \ref{subsec:discrete_generalization}. In what follows, $\theta \in [0,1]$ is a sensitivity parameter that indexes relaxations of exogeneity or weak exclusion of the instrument.

\begin{het1Assump}[General Sensitivity Model]\label{assn:ind_relax}
 For a known sensitivity parameter $\theta \in [0,1]$, let
\begin{align*}
 \textbf{p}_Y \in \mathcal{A}_0(\theta) \times \mathcal{A}_1(\theta)
\end{align*}
where, for $x \in \{0,1\}$, $\mathcal{A}_x$ satisfies
\begin{enumerate}
 \item(Spanning) $\mathcal{A}_x(0) = \{a \in [0,1]^2: a_0 = a_1\}$ and $\mathcal{A}_x(1) = [0,1]^2$;
 \item(Monotonicity) $\mathcal{A}_x(\theta) \subseteq \mathcal{A}_x(\theta')$ when $\theta \leq \theta'$;
 \item(Linearity of Constraints)  $\mathcal{A}_x(\theta)$ is a closed convex polytope for each $\theta \in [0,1]$;
 \item(Continuity) The correspondence $\mathcal{A}_x: [0,1] \rightrightarrows [0,1]^2$ is continuous.
\end{enumerate}
\end{het1Assump}

The first part of this assumption implies that setting $\theta = 0$ imposes exogeneity and weak exclusion of the instrument, while setting $\theta = 1$ implies no restrictions on the dependence between $Z$ and the potential outcomes. The second part assumes these restrictions are monotonic in $\theta$, meaning that increasing $\theta$ yields a (weakly) larger set of conditional probabilities. These two parts combined yield that $\{\mathcal{A}_x(\theta):\theta \in [0,1]\}$ monotonically connects no assumptions to exogeneity and weak exclusion. The third restriction says that these sets are characterized by finitely many weak linear inequalities. This is crucial in obtaining a linear programming formulation for the bounds of various causal objects, such as the ATE.  The last part of this assumption assumes the continuity of the correspondence between the sensitivity parameter $\theta$ and the set of restricted conditional probabilities. Recall that a correspondence is continuous if it is both upper and lower hemicontinuous (uhc and lhc) at all points of its domain. See \cite{Border1985} for a compendium of results related to continuity of correspondences we make use of in our proofs. This assumption will yield continuity in the sensitivity parameter of the causal bounds obtained from linear programming.

This high-level assumption has useful properties, and all three previously considered relaxations are special cases of it. This is formalized in this proposition.
\begin{proposition}\label{prop:high-level_relax_discrete}
 Suppose Assumption \ref{assn:NonTrivialinstrument} holds. Relabeling $\lambda$, $c$, and $K$ as $\theta \in [0,1]$, the sets $\mathcal{A}_\text{MSM}(\lambda)$, $\mathcal{A}_{c\text{-dep}}(c)$, and $\mathcal{A}_\text{KS}(K)$ satisfy Assumption \ref{assn:ind_relax}.
\end{proposition}

Under this general relaxation, we will derive identified sets for various parameters of interest. We use these identified sets to characterize sharp bounds on causal objects using linear programming. We can also use them to determine what values of $\theta$ correspond to falsified models.

Before continuing our discussion, we present the identified set for conditional outcome probabilities under this general restriction.
\begin{theorem}\label{thm:prob_IDset}
 Suppose assumptions \ref{assn:NonTrivialinstrument}, \ref{assn:exog_or_excl}, and \ref{assn:ind_relax} hold. Then:
 \begin{enumerate}
  \item The identified set for $\textbf{p}_Y$ is
\begin{align}\label{eq:IDset_corr}
 \Pi(\theta) \coloneqq \Pi_0(\theta) \times \Pi_1(\theta),
\end{align}
where $\Pi_x(\theta) \coloneqq \mathcal{H}_x \cap \mathcal{A}_x(\theta)$;
  \item There exists $\underline{\theta}\in [0,1]$ such that $\Pi(\theta)$ is non-empty for $\theta \geq \underline{\theta}$ and empty for $\theta < \underline{\theta}$;
  \item For all $\theta \in [\underline{\theta},1]$, $\Pi(\theta)$ is a closed convex polytope in $[0,1]^4$;
  \item For $x \in \{0,1\}$, suppose $\text{int}(\mathcal{H}_x \cap \mathcal{A}_x(\theta)) \neq \emptyset$ for all $\theta > \underline{\theta}$. The correspondence $\Pi:[\underline{\theta},1] \rightrightarrows [0,1]^4$ defined in equation \eqref{eq:IDset_corr} is continuous.
 \end{enumerate}
\end{theorem}
This theorem has several implications. First, the identified set for the set of probabilities $\textbf{p}_Y$ is a Cartesian product of two sets. Each of these two sets is characterized as the intersection between a set containing all vectors $\textbf{p}_Y(x)$ consistent with the distribution of observables $F_{Y,X,Z}$, and the set of vectors consistent with a sensitivity model indexed by $\theta$.

The second implication is that the sensitivity model is falsified for an open, but potentially empty, subset of $[0,1]$. The minimum value at which the model is not falsified, called the \textit{falsification point} by \cite{MastenPoirier2021}, is $\underline{\theta}$ and is identified since it is a property of the sets $\Pi(\theta)$ for $\theta \in [0,1]$, all of which are known from the distribution of $(Y,X,Z)$. Moreover, the set of values for which the identified set is non-empty is closed, and always contains $\theta = 1$.

Third, this set is a closed, convex polytope, meaning it is defined by finitely many linear inequalities. This ensures that optimizing linear functions, such as $\Gamma_\text{ATE}$ and $\Gamma_\text{ATT}$ from \eqref{eq:ATE_functional} and \eqref{eq:ATT_functional}, can be performed using linear programming. This will be the key computational tool for implementing these methods.

Fourth, and finally, the mapping from $\theta$ into the identified set is continuous as a correspondence. This allows us to show the continuity in $\theta$ of extrema of continuous functionals of $\textbf{p}_Y$ over the identified set, again including the ATE and ATT. 

We now illustrate the identified set for a sensitivity model corresponding to $c$-dependence. The shaded boxes in Figure \ref{fig:FFbinaryYsingleIV} show examples of the no-assumption bounds $\mathcal{H}_0 \times \mathcal{H}_1$.

\begin{figure}[t]
\centering
\includegraphics[width=77mm]{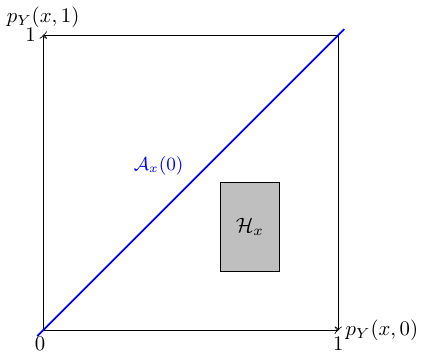}
\hspace{5mm}
\includegraphics[width=77mm]{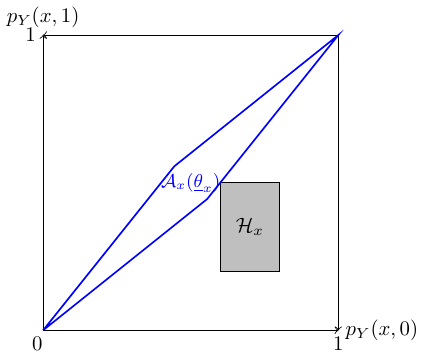}\\
\includegraphics[width=77mm]{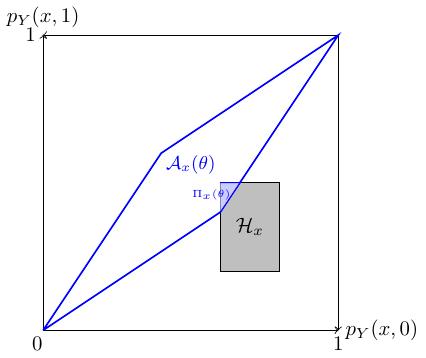}
\hspace{5mm}
\includegraphics[width=77mm]{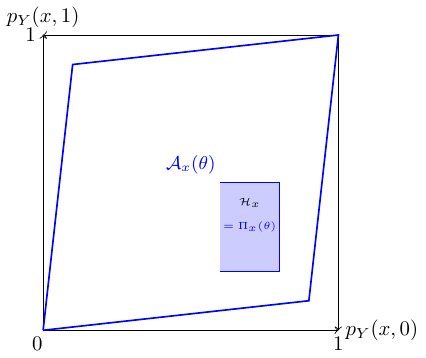}
\footnotesize
\caption{Example identified sets for different sensitivity parameter values. Relaxation is $c$-dependence. Top Left: $\theta = 0$. Top Right: $\theta = \underline{\theta}_x$, the falsification point. Bottom left: $\theta > \underline{\theta}_x$. Bottom right: $\theta$ sufficiently large that the identified set equals the no-assumption identified set from Proposition \ref{prop:noassn_discrete}. }
\label{fig:IDset_sensparam}
\end{figure}

The set $\mathcal{A}_x(\theta)$ is a parallelogram imposing the $c$-dependence constraint. The identified set for $\textbf{p}_Y(x)$ is given by the intersection of the parallelogram and shaded box.

While the no-assumption bounds are never empty, the bounds under exogeneity and weak exclusion ($\theta=0$) can be empty, and hence the baseline statistical independence assumption can be falsified. This happens when, for some $x \in \{0,1\}$, the no assumption bounds $\mathcal{H}_x$ have an empty intersection with the statistical independence constraint set $\mathcal{A}_x(\theta)$. Graphically, this happens when the box defined by the no assumption bounds does not intersect the 45-degree line. This is shown in the first plot Figure \ref{fig:IDset_sensparam}. The falsification point is simply the smallest value of $\theta$ such that the parallelogram defined by $\mathcal{A}_x(\theta)$ has a nonempty intersection with the no assumption bounds $\mathcal{H}_x$ for each $x \in \{0,1\}$. This intersection is illustrated in the second plot of Figure \ref{fig:IDset_sensparam}. Increasing the sensitivity parameter increases the size of this intersection (see the third plot of Figure \ref{fig:IDset_sensparam}), until the intersection equals $\mathcal{H}_x$, the no assumption bounds, which can be seen in the fourth plot of Figure \ref{fig:IDset_sensparam}.

We next show how to use Theorem \ref{thm:prob_IDset} to get identified sets for  counterfactual probabilities $\Prob(Y(x)=1)$ and for the ATE. By the law of total probability,
\[
 \Prob(Y(x)=1) = p_Y(x,0) (1-p_Z) + p_Y(x,1) p_Z.
\]
The weight $p_Z$ is identified, while the identified set for $\textbf{p}_Y(x)$ is given by $\Pi_x(\theta)$. Thus, we can simply minimize and maximize the above convex combination over this set to obtain the identified set for $\Prob(Y(x)=1)$. Hence we define
\begin{align*}
	\overline{P}_x(\theta) &\coloneqq \sup_{(a_0,a_1) \in \Pi_x(\theta)} \big( a_0 (1 - p_Z) + a_1 p_Z \big) \qquad \text{and} \qquad \underline{P}_x(\theta) &\coloneqq \inf_{(a_0,a_1) \in \Pi_x(\theta)} \big( a_0 (1 - p_Z) + a_1 p_Z \big).
\end{align*}
These are both finite-dimensional linear programs and hence can be computed easily given estimates of the joint distribution of $(Y,X,Z)$. Figure \ref{fig:IDset_sensparam_optima} illustrates the minimization/maximization of a linear functional over the identified set $\mathcal{A}_x(\theta) \cap \mathcal{H}_x$. 

\begin{figure}[t]
\centering
\includegraphics[width=105mm]{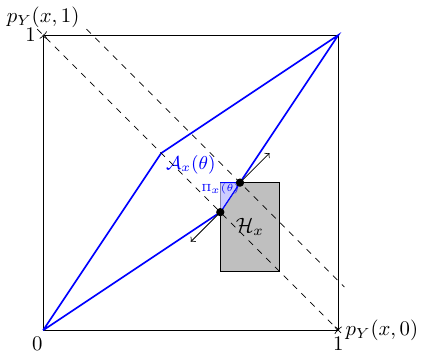}
\footnotesize
\caption{Example identified set and minimization/maximization of the ATE under a sensitivity model.}
\label{fig:IDset_sensparam_optima}
\end{figure}

The following corollary lists properties of these bounds.
\begin{corollary}\label{corr:hetTrtBinATEbounds}
Suppose assumptions \ref{assn:NonTrivialinstrument}, \ref{assn:exog_or_excl}, and \ref{assn:ind_relax} hold. Then, for $x \in \{0,1\}$:
\begin{enumerate}
\item The identified set for $(\Prob(Y(0)=1), \Prob(Y(1)=1))$ is $I_0(\theta) \times I_1(\theta)$ where $I_x \coloneqq [\underline{P}_x(\theta), \overline{P}_x(\theta)]$ when $\theta\in [\underline{\theta},1]$, and the empty set when $\theta < \underline{\theta}$;

\item The functions $\underline{P}_x(\theta)$ and $\overline{P}_x(\theta)$ are continuous and monotonic over $\theta \in [\underline{\theta},1]$;

\item Let $\theta \in [\underline{\theta},1]$. The identified set for ATE is $[\underline{\text{ATE}}(\theta),\overline{\text{ATE}}(\theta)]$ where 
\[
 \underline{\text{ATE}}(\theta) \coloneqq \underline{P}_1(\theta) - \overline{P}_0(\theta)
 \qquad \text{and} \qquad
 \overline{\text{ATE}}(\theta) \coloneqq \overline{P}_1(\theta) - \underline{P}_0(\theta).
\]
\end{enumerate}
\end{corollary}

 This discussion implies that ATE will typically be partially identified at the falsification point. That is: The falsification adaptive set for ATE, $[\underline{\text{ATE}}(\underline\theta), \overline{\text{ATE}}(\underline\theta)]$, will generally be an interval with a nonempty interior.

\subsection{Generalization to Non-Binary Discrete Variables}\label{subsec:discrete_generalization}
The previous results illustrate that common sensitivity models yield identified sets for parameters of interest with desirable properties. However, these were illustrated only for cases where $Y$, $X$, and $Z$ were all binary. In practice, many empirical settings have multiple instruments, and treatments or outcomes may be multivalued as well. In this section, we sketch a generalization of the previous results to cases where the support of $Y$, $X$, may be discrete instead of binary, where there may be multiple instruments, and where each instrument may possess finite support rather than being binary.

Let $X$ be a discrete treatment, let $Z$ be a vector of instruments, where each instrument is discrete, and let $Y(x,z)$ be discrete as well. We let $Y = Y(X,Z)$ be the realized, observed outcome. We suppose all their supports are finite.

Let $p_Y(y \mid z ; \ x) \coloneqq \Prob(Y(x) = y\mid Z = z)$, $\textbf{p}_{Y}(x,z) \coloneqq \{p_Y(y \mid z ; \ x)\}_{y \in \supp(Y(x))}$, $\textbf{p}_{Y}(x) \coloneqq \{\textbf{p}_{Y}(x,z)\}_{z \in \supp(Z)}$, and $\textbf{p}_{Y} \coloneqq \{\textbf{p}_{Y}(x)\}_{x \in \supp(X)}$. The vector $\textbf{p}_{Y}$ contains the full distribution of $Y(x) \mid \{Z = z\}$ for all $(x,z) \in \supp(X,Z)$. We define $s_Z \coloneqq |\supp(Z)|$, $s_{Y(x)} \coloneqq |\supp(Y(x))|$, and $\supp(Y(x)) \coloneqq \{y_1,\ldots,y_{s_{Y(x)}}\}$.

To avoid trivial cases, we make the following assumption.
\begin{het1Assump}\label{assn:NonTrivialinstrument_disc}
For all $(x, z) \in \supp(X) \times\supp(Z)$, $\P(Z = z) \in (0,1)$ and $\Prob(X=x \mid Z =z) \in (0,1)$.
\end{het1Assump}

Let $\Delta_{S}$ denote the simplex of dimension $S$:
\begin{align*}
 \Delta_{S} &\coloneqq \left\{\textbf{p} \in \R^{S+1}: \textbf{p} \geq 0, \sum_{s = 1}^{S+1} p_s = 1\right\}.
\end{align*}
For $K \in \mathbb{N}$, let $\Delta_S^K$ denote the $K$-fold cartesian product of $\Delta_S$. 

The no-assumption identified set for $\textbf{p}_Y$ is given by 
\begin{align*}
 \prod_{x \in \supp(X)} \prod_{z \in \supp(Z)}  \mathcal{H}_{x,z},
\end{align*}
where
\begin{align*}
 \mathcal{H}_{x,z} &= \{ (p_1,\ldots,p_{s_{Y(x)}}) \in \Delta_{s_{Y}(x) - 1}: p_s \in [\Prob(Y=y_s,X=x\mid Z = z),\Prob(Y=y_s,X=x\mid Z = z) \\
 &\qquad + \Prob(X \neq x\mid Z = z)], \text{for all } s \in \{1,\ldots,s_{Y(x)}\}\}.
\end{align*}
The set $\mathcal{H}_{x,z}$ is the identified set for $\textbf{p}_Y(x,z)$ under no assumptions. We also note the similar structure of $\mathcal{H}_x = \prod_{z \in \supp(Z)}  \mathcal{H}_{x,z}$  and of the rectangles defined in \eqref{eq:boxSet} for the binary case.

The three sensitivity models we investigated earlier can be defined independently of the supports of the potential outcomes, treatments, or instruments, so they can be used when these variables are non-binary. We can also embed these sensitivity models in a general sensitivity model similar to the one in Assumption \ref{assn:ind_relax}. The following assumption simplifies to \ref{assn:ind_relax} when all variables are binary.

\begin{het1Assump}[General Sensitivity Model]\label{assn:ind_relax_discrete}
 Suppose Assumption \ref{assn:NonTrivialinstrument_disc} holds. For a known sensitivity parameter $\theta \in [0,1]$, let
\begin{align*}
 \textbf{p}_Y \in \prod_{x \in \supp(X)} \mathcal{A}(\theta; x)
\end{align*}
where, for $x \in \supp(X)$, $\mathcal{A}(\theta; x)$ satisfies
\begin{enumerate}
 \item(Spanning) $\mathcal{A}(0; x) = \{ (a_1,\ldots,a_{s_Z}) \in \Delta_{s_Y(x) - 1}^{s_Z} : a_1 = \cdots = a_{s_Z}\}$ and $\mathcal{A}(1; x) = \Delta_{s_Y(x) - 1}^{s_Z}$;
 \item(Monotonicity) $\mathcal{A}(\theta; x) \subseteq \mathcal{A}(\theta'; x)$ when $\theta \leq \theta'$;
 \item(Linearity of Constraints)  $\mathcal{A}(\theta; x)$ is a closed convex polytope for each $\theta \in [0,1]$;
 \item(Continuity) The correspondence $\mathcal{A}(\cdot; x): [0,1] \rightrightarrows \Delta_{s_Y(x) - 1}^{s_Z}$ is continuous.
\end{enumerate}
\end{het1Assump}
This assumption is similar to its counterpart with binary variables, except for parts 1 and 4, which have been modified to allow $Y(x)$ to be nonbinary. The restriction in part 1 states that $\P(Y(x) = y \mid  Z = z)$ is constant in $z \in \supp(Z)$ for each $y \in \supp(Y(x))$, and is stated as equality constraints for on the components of $\mathcal{A}_x(0)$.

As in the binary case, all these assumptions can be written as linear inequalities in the components of vector $\textbf{p}_Y$. Therefore, the bounds on various causal objects can be obtained by solving linear programs. We expect similar results to Theorem \ref{thm:prob_IDset} and Corollary \ref{corr:hetTrtBinATEbounds} to hold in this setting, so that the bounds enjoy the same monotonicity and continuity property.

\section{Identification with Continuous Outcomes}\label{sec:cont_Y}

We now consider cases where the outcome variable is continuously distributed. In this case, we may view the corresponding problem as an infinite dimensional program, whose theoretical properties are harder to analyze. Nevertheless, in this section we show that the previous sensitivity models can be used with continuous outcomes, and we obtain theoretical properties of the corresponding sensitivity analyses for the exogeneity/exclusion of an instrument. To keep other aspects of the problem relatively simple, we consider the case where the treatment and instrument are both binary, although this can be naturally generalized as in Section \ref{subsec:discrete_generalization}. In this section, we show that the analytical results we derived under binary outcomes generalize to continuous outcomes. This leads us to a relatively simple and feasible approach for computing identified sets under relaxations of instrument exogeneity with continuous outcomes.

We begin by assuming that outcomes are continuously distributed.

\begin{het1Assump}
\label{assn:contsupport}
Suppose that $\supp(Z) = \supp(X) = \{0,1\}$. For any $x, x',z \in \{0,1\}$ the distribution of $Y(x) \mid \{X=x', Z=z\}$ is continuous with respect to the Lebesgue measure and is supported on a compact interval $\mathcal{Y}_x \coloneqq \supp(Y(x))$, which is independent of $x'$ and $z$.
\end{het1Assump}

Assumption \ref{assn:contsupport} supposes that, conditional on the treatment and instruments, potential outcomes are continuously distributed. It implies that, conditional on the treatment and instruments, observed outcomes are also continuously distributed. We can allow for discrete instruments as in Section \ref{sec:binary_outcome}, but we only consider a binary instrument to simplify the notation. 

This assumption also states that the conditional support of $Y(x)$ given $(X,Z) = (x',z)$ does not depend on $(x',z)$, which is made for convenience. Our results would remain valid without this restriction, but notation in the proofs would have to be heavier.

Let $f_{Y}(y \mid z; \ x) \coloneqq f_{Y(x) \mid Z}(y \mid z)$ denote the conditional density of $Y(x)$ given $Z = z$. We also let $\textbf{f}_{Y}(y ; \ x) \coloneqq (f_{Y}(y \mid 0 ; \ x),f_{Y}(y \mid 1 ; \ x))$ and  $\textbf{f}_{Y}(y) \coloneqq (\textbf{f}_{Y}(y ; \ 0), \textbf{f}_{Y}(y ; \ 1))$ denote collections of these densities across instrument and treatment values.  We assume that the potential outcomes' densities belong to a convex class of densities that is compact with respect to the supremum norm.
\begin{het1Assump}\label{assn:compactdensities}
For $x,z\in\{0,1\}$, let
\begin{align*}
	f_{Y}(\cdot \mid z ; \ x) &\in \left\{f \in \mathcal{F}_x(\mathcal{Y}_x): \int_{\mathcal{Y}_x}f(y) \, dy = 1, f \geq 0\right\} \eqqcolon \mathcal{F}_{\text{den},x},
\end{align*}
where $\mathcal{F}_x(\mathcal{Y}_x)$ is a convex set of bounded functions supported on $\mathcal{Y}_x$ that is compact with respect to the norm $\|f\|_\infty \coloneqq \sup_{y \in \mathcal{Y}_x}|f(y)|$.
\end{het1Assump}

Examples of compact sets $\mathcal{F}_x(\mathcal{Y}_x)$  include the set of bounded Lipschitz functions:
\begin{align*}
	\mathcal{C}_{0,\infty,1,1}(\mathcal{Y}_x) \coloneqq \left\{f \in \mathcal{C}_0(\mathcal{Y}_x): \sup_{y \in \mathcal{Y}_x}|f(y)| + \sup_{y,y' \in \text{int}(\mathcal{Y}_x), y\neq y'} \frac{|f(y') - f(y)|}{|y' - y|} <M \right\}
\end{align*}
where $\mathcal{C}_0(A)$ denotes the set of continuous functions on domain $A$, and $M<\infty$ is a constant. See \cite{FreybergerMasten2019} for alternative compact sets of functions and associated discussion. 

We start by deriving the no-assumptions bounds for this set of conditional densities.
\begin{proposition}\label{prop:cont_noassbounds}
Suppose assumptions \ref{assn:NonTrivialinstrument}, \ref{assn:contsupport}, and \ref{assn:compactdensities} hold. The identified set for $\textbf{f}_{Y}$ is 
\begin{align*}
	\mathcal{H} \coloneqq \prod_{x = 0,1} \mathcal{H}_{x}
\end{align*} where $\mathcal{H}_{x} \coloneqq \prod_{z = 0,1} \mathcal{H}_{x,z}$ and $\mathcal{H}_{x,z} \coloneqq \{f(\cdot) \in \mathcal{F}_{\text{den},x}: f \geq f_{Y|X,Z}(\cdot\mid x,z)\pi(x\mid z)\}$.
\end{proposition}

We next consider the baseline case where the instruments are exogenous and excluded. In this case, the instrument's validity implies that the densities $\textbf{f}_Y$ must lie in
\begin{align*}
	\mathcal{A}_\text{indep} \coloneqq \{(f_{00},f_{01},f_{10},f_{11}) \in \mathcal{F}_0(\mathcal{Y}_0)^2 \times \mathcal{F}_1(\mathcal{Y}_1)^2 : f_{00} = f_{01}, f_{10} = f_{11}\},
\end{align*}
since this set imposes that $f_{Y(x)|Z}(\cdot\mid 0) = f_{Y(x)|Z}(\cdot\mid 1)$ for $x = 0,1$. Thus, the identified set for $\textbf{f}_Y$ under independence is given by
\begin{align*}
	\mathcal{H} \cap \mathcal{A}_\text{indep}.
\end{align*}
This is precisely the setting studied in \cite{Kitagawa2021}, and he provides a characterization of this set in his Proposition 3.1, which we include without proof.
\begin{proposition}\label{prop:kitagawa3point1} (Proposition 3.1 in \cite{Kitagawa2021}) Suppose assumptions \ref{assn:NonTrivialinstrument}, \ref{assn:contsupport}, and \ref{assn:compactdensities} hold. Suppose $Z$ is exogenous and weakly excluded. Then the identified set for $(f_{Y(0)},f_{Y(1)})$ is 
\[
	 \left\{ f_0 \in \mathcal{F}_{\text{den},0}: f_0(\cdot) \geq \max_{z=0,1} \; f_{Y|X,Z}(\cdot\mid  0 , z) \pi(0\mid z)\right\} \times \left\{ f_1 \in \mathcal{F}_{\text{den},1}: f_1(\cdot) \geq \max_{z=0,1} \; f_{Y|X,Z}(\cdot \mid 1, z)\pi(1\mid z)\right\}.
\]
Consequently, the model is refuted if
\[
	\int_{\mathcal{Y}_x} \max_{z=0,1} \; f_{Y,X|Z}(y,x \mid z) \; dy > 1
\]
for some $x \in \{0,1 \}$.
\end{proposition}

The previous two results establish the identification region for conditional densities of $Y(x)$ given $Z$ under no-assumptions, and under the full validity of the instrument, which correspond to the ends of a spectrum of assumptions about the dependence between $Y(x)$ and $Z$. We now consider sensitivity models that consider intermediate assumptions on the instrument's validity. We again consider the following three restrictions, which are adapted from Section \ref{subsec:partial_indep_assn}.

\subsubsection*{Marginal Sensitivity Model}
Consider the Marginal Sensitivity Model of definition \ref{def:MSModel}. When the outcome is continuously distributed, Bayes' rule allows us to rewrite equation \eqref{eq:MSModel} as a density ratio:
\begin{align*}
	 \frac{f_{Y}(\cdot\mid z; \ x)}{f_{Y}(\cdot\mid z'; \ x)} \in \left[\Lambda^{-1},\Lambda\right]
\end{align*}
for $(x,z,z') \in \supp(X) \times \supp(Z)^2$. As in the previous sections, we reparametrize $\Lambda$ as $\lambda = 1 - \Lambda^{-1} \in [0,1]$. The set of densities satisfying this restriction can be viewed as a set of functions satisfying linear inequality constraints. Specifically, we can write the set of restricted densities as
\begin{align}\label{eq:MSM_cont}
	\mathcal{A}_\text{MSM}(\lambda) &= \mathcal{A}_\text{MSM}(\lambda;0) \times \mathcal{A}_\text{MSM}(\lambda;1)
\end{align}
where $\mathcal{A}_\text{MSM}(\lambda;x) \coloneqq \{\textbf{f} \in \mathcal{F}_{\text{den},x}^2: A_\text{MSM}(\lambda)\textbf{f} \leq \textbf{0}\}$ and
\begin{align*}
	A_\text{MSM}(\lambda) \coloneqq \begin{pmatrix}
		-1 & 1-\lambda\\
		1-\lambda & -1
	\end{pmatrix}.
\end{align*} 
Inequalities involving functions $\textbf{f}$ are meant to hold across all $y \in \R$.

\subsubsection*{$c$-dependence}
As defined in equation \eqref{eq:c-indep1}, $c$-dependence is collection of inequalities across values of $y$. Again using Bayes' Rule, we can rewrite these inequalities using conditional densities of $Y(x)$ given the instrument:
\begin{align*}
	-\min\{p_Z + c,1\}(1-p_Z)f_{Y}(y\mid 0;x) + (1 - \min\{p_Z + c,1\})p_Z f_{Y}(y\mid 1;x) &\leq 0\\
	\max\{p_Z - c,0\}(1-p_Z)f_{Y}(y\mid 0;x) + (\max\{p_Z - c,0\} - 1)p_Z f_{Y}(y\mid 1;x) &\leq 0.
\end{align*}

These are densities restricted by linear inequalities that depend on the observed variables only through the marginal distribution of the instrument. The set of densities as restricted by $c$-dependence is given by
\begin{align}\label{eq:cdep_cont}
	\mathcal{A}_\text{$c$-dep}(c) &\coloneqq \mathcal{A}_\text{$c$-dep}(c;0) \times \mathcal{A}_\text{$c$-dep}(c;1)
\end{align}
where $\mathcal{A}_\text{$c$-dep}(c;x) \coloneqq \{\textbf{f} \in \mathcal{F}_{\text{den},x}^2: A_\text{$c$-dep}(c)\textbf{f} \leq \textbf{0}\}$ and
\begin{align}\label{eq:cdep_cont_matrix}
		A_\text{$c$-dep}(c) &\coloneqq \begin{pmatrix}
		-1 & k_1(c)\\
		k_0(c) & -1
	\end{pmatrix}.
\end{align}
We can see that setting $c = 0$ implies that $k_0(c) = k_1(c) = 1$, which mechanically imposes that the conditional densities $f_{Y(x)|Z}(\cdot\mid 0)$ and $f_{Y(x)|Z}(\cdot\mid 1)$ are equal. As a result, we can verify that $c=0$ implies independence of potential outcomes and the instrument, as it does when the outcome is discrete.

\subsubsection*{Supremum Distance}
Using the Kolmogorov-Smirnov as a starting point, we consider a sensitivity model that bounds the supremum distance between densities rather than distribution functions. Hence, we assume that 
\begin{align*}
	\sup_{y \in \R}|f_Y(y\mid 0;x) - f_{Y}(y\mid 1;x)| \leq K/(1-K)
\end{align*}
for $x \in \supp(X)$, for some known $K$ satisfying $K \in [0,1]$.\footnote{We let $K/(1-K) = +\infty$ when $K = 1$.} The sensitivity parameter $K$ bounds the difference between density functions, and we used the strictly increasing mapping $a\mapsto 1/(1-a)$ to span the continuum between independence and no restrictions, as $K=0$ maps to exact equality of densities, and $K=1$ does not impose any restrictions on the dependence of the distribution of $Y(x)$ in $Z$. An alternate mapping from $[0,1]$ to $[0,+\infty]$ could be used instead.

The set of densities as restricted by this sup distance is given by
\begin{align}\label{eq:KS_cont}
	\mathcal{A}_\text{KS}(K) &\coloneqq \mathcal{A}_\text{KS}(K;0) \times \mathcal{A}_\text{KS}(K;1)
\end{align}
where $\mathcal{A}_\text{KS}(K;x) \coloneqq \{\textbf{f} \in \mathcal{F}_{\text{den},x}^2: A_\text{KS}\textbf{f} \leq (K/(1-K), K/(1-K))^\top\}$, where $A_\text{KS}$ is defined as in equation \eqref{eq:KS_Matrix}.

\subsection{A Unifying Sensitivity Model with Continuous Outcomes}

As in Section \ref{sec:binary_outcome}, all these relaxations can be viewed as special cases of a unifying class of relaxations encoding various types of departures from independence.

\begin{het1Assump}[General Sensitivity Model with Continuous Outcomes]\label{assn:ind_relax_cont}
	For a known sensitivity parameter $\theta \in [0,1]$, suppose
\begin{align*}
	\textbf{f}_Y \in \mathcal{A}_0(\theta) \times \mathcal{A}_1(\theta)
\end{align*}
where, for $x \in \{0,1\}$, $\mathcal{A}_x$ satisfies
\begin{enumerate}
	\item(Spanning) $\mathcal{A}_x(0) = \{(f_0,f_1) \in \mathcal{F}_{\text{den},x}^2: f_0 = f_1\}$ and $\mathcal{A}_x(1) = \mathcal{F}_{\text{den},x}^2$;
	\item(Monotonicity) $\mathcal{A}_x(\theta) \subseteq \mathcal{A}_x(\theta')$ when $\theta \leq \theta'$;
	\item(Linearity of Constraints)  The set $\mathcal{A}_x(\theta)$ is a closed convex subset of $\mathcal{F}_{\text{den},x}^2$ characterized by finitely many componentwise weak linear inequalities in the densities for each $\theta \in [0,1]$;
	\item(Continuity) The correspondence $\mathcal{A}_x: [0,1] \rightrightarrows \mathcal{F}_{\text{den},x}^2$ is continuous with respect to the sup-norm.
\end{enumerate}
\end{het1Assump}
The constraint set $\mathcal{A}_0(\theta) \times \mathcal{A}_1(\theta)$ is a convex set of functions defined by linear inequalities that weakly expands as $\theta$ increases. It nests the identified set under the baseline independence assumption ($\theta=0$) and the identified set under no assumptions on the dependence between potential outcomes and instruments ($\theta=1$). The third requirement is that the constraint set is of the form $\mathcal{A}_x(\theta) = \{\textbf{f} = (f_0,f_1) \in \mathcal{F}_{\text{den},x}^2: A(\theta)\textbf{f} \leq a(\theta)\}$ where $A(\theta)$ is a finite dimensional matrix. It involves finitely many componentwise weak inequalities, even though the inequality $A(\theta)\textbf{f} \leq a(\theta)$ hold for infinitely many values on the support $\mathcal{Y}_x$.   As in the binary outcome case, this relaxation encompasses the previous three restrictions.

\begin{proposition}\label{prop:high-level_relax_cont}
	Suppose assumptions \ref{assn:NonTrivialinstrument}, \ref{assn:contsupport}, and \ref{assn:compactdensities} hold. Relabeling $(\lambda,K,c)$ as $\theta \in [0,1]$, the restrictions from equations \eqref{eq:MSM_cont}, \eqref{eq:cdep_cont}, and \eqref{eq:KS_cont} all satisfy Assumption \ref{assn:ind_relax_cont}.
\end{proposition}

We now state our main result about theoretical properties of the identified set for densities of the potential outcomes.
\begin{theorem}\label{thm:cont_prob_IDset}
	Suppose assumptions \ref{assn:NonTrivialinstrument}, \ref{assn:contsupport}, \ref{assn:compactdensities}, and \ref{assn:ind_relax_cont} hold, and suppose that $\mathcal{F}_{\text{den},x}^2 \cap \mathcal{H}_x \neq \emptyset$ for $x \in \{0,1\}$. Then,
	\begin{enumerate}
		\item The identified set for $\textbf{f}_Y$ is
\begin{align}\label{eq:cont_IDset_corr}
	\Pi(\theta) \coloneqq \Pi_0(\theta) \times \Pi_1(\theta),
\end{align}
where
\begin{align*}
	\Pi_x(\theta) &\coloneqq \mathcal{H}_x \cap \mathcal{A}_x(\theta);
\end{align*}
		\item There exists $\underline{\theta}\in [0,1]$ such that $\Pi(\theta)$ is non-empty for $\theta \geq \underline{\theta}$ and empty for $\theta < \underline{\theta}$;
		\item Suppose $\text{int}(\mathcal{H}_x \cap \mathcal{A}_x(\theta)) \neq \emptyset$ for all $\theta > \underline{\theta}$. Then, the correspondence $\Pi:[\underline{\theta},1] \rightrightarrows \mathcal{F}_{\text{den},0}^2 \times \mathcal{F}_{\text{den},1}^2$ defined by $\Pi(\theta)$ in equation \eqref{eq:cont_IDset_corr} is continuous.
	\end{enumerate}
\end{theorem}

This theorem establishes the main theoretical properties of the identified sets for densities, including their continuity as an infinite-dimensional correspondence. This continuity will carry over to functionals of these densities, in particular to linear or continuous functionals.

In particular, consider the class of linear mappings, for which the sharp bounds can be obtained as the solution to a linear program.
Let
\[
	\Gamma(\textbf{f}) 
	\coloneqq \int_{\mathcal{Y}_0} \omega_0(y)'\textbf{f}_Y(y ; \ 0)dy + \int_{\mathcal{Y}_1} \omega_1(y)'\textbf{f}_Y(y ; \ 1)dy
\]
where, for $x = 0,1$, $\omega_x$ is a known weight function that maps $\R$ to $\R^2$. The $\Gamma$ mapping is used to characterize a functional of the conditional densities of $Y(x)\mid Z$.

For example, with $\omega_1(y) = -\omega_0(y) = (y (1-p_Z), y p_Z)$, we have that
\begin{align*}
	\Gamma(\textbf{f}_Y) &= \int_{\mathcal{Y}_1} y \left(p_Zf_{Y}(y\mid 1; \ 1) + (1-p_Z)f_{Y}(y\mid 0 ; \ 1)\right) dy\\
	&\qquad - \int_{\mathcal{Y}_0} y \left(p_Z f_{Y}(y\mid 1; \ 0) + (1-p_Z)f_{Y}(y\mid 0 ; \ 0)\right) dy\\
	&= \int_{\mathcal{Y}_1} y f_{Y(1)}(y) dy - \int_{\mathcal{Y}_0} y f_{Y(0)}(y) dy\\
	&= \Exp[Y(1)] - \Exp[Y(0)],
\end{align*}
the average treatment effect. Letting $\omega_x(y) = \indicator(y \leq a)$ and $\omega_{1-x}(y) = 0$ yields $\Gamma(\textbf{f}_Y) = \P(Y(x) \leq a)$, the cumulative distribution function evaluated at $a \in \R$.. This choice can be used to obtain bounds on quantiles of $Y(x)$ or on the quantile treatment effect $\text{QTE}(\tau) \coloneqq Q_{Y(1)}(\tau) - Q_{Y(0)}(\tau)$ for a quantile index $\tau \in (0,1)$.

The proposition below shows that bounds on these functionals are continuous and monotonic. This result uses the Maximum Theorem \citep{Berge1959} applied to an infinite-dimensional correspondence. 
Let
\[
	\overline{\Gamma}(\theta)
	\coloneqq \sup_{\textbf{f}_1 \in \Pi_1(\theta)} \int_{\mathcal{Y}_1} \omega_1(y_1)'\textbf{f}_1(y_1) \; dy_1 
	+ \sup_{\textbf{f}_0 \in \Pi_0(\theta)}\int_{\mathcal{Y}_0} \omega_0(y_0)'\textbf{f}_0(y_0) \; dy_0
\]
and
\[
	\underline{\Gamma}(\theta) 
	\coloneqq \inf_{\textbf{f}_1 \in \Pi_1(\theta)} \int_{\mathcal{Y}_1} \omega_1(y_1)'\textbf{f}_1(y_1) \; dy_1 
	+ \inf_{\textbf{f}_0 \in \Pi_0(\theta)}\int_{\mathcal{Y}_0} \omega_0(y_0)'\textbf{f}_0(y_0) \; dy_0
\]
denote the lower and upper bounds of the functional over the sets $\Pi_x(\theta)$, $x=0,1$.
\begin{corollary}\label{corr:cont_functionalbounds}
Suppose the assumptions of Theorem \ref{thm:cont_prob_IDset} hold. Let $\|\omega_x(\cdot)\|_\infty < \infty$. Then,
\begin{enumerate}
\item Let $\theta\in [\underline{\theta},1]$. The identified set for $(f_{Y(0)}, f_{Y(1)})$ is $I_0(\theta) \times I_1(\theta)$ where $I_x(\theta) \coloneqq \{f_0 (1-p_Z) + f_1 p_Z: (f_0,f_1) \in \Pi_x(\theta)\}$ when $\theta\in [\underline{\theta},1]$, and the empty set when $\theta < \underline{\theta}$;

\item The functions $\underline{\Gamma}(\theta)$ and $\overline{\Gamma}(\theta)$ are continuous and monotonic over $\theta \in [\underline{\theta},1]$.

\item Let $\theta \in [\underline{\theta},1]$. The identified set for $\Gamma(\textbf{f}_{Y})$ is $[\underline{\Gamma}(\theta), \overline{\Gamma}(\theta)]$.
\end{enumerate}
\end{corollary}

Therefore, as in the discrete case, bounds a can be obtained in the continuous case through infinite-dimensional linear programming. To make this approach feasible, we show in the next section how to convert an infinite-dimensional linear program into a feasible, finite-dimensional linear program that can be directly implemented.

\subsection{Computation}\label{sec:comp}

The identified set $\Pi(\theta)$ is an infinite-dimensional set of continuous densities. If we restrict attention to the class of linear functionals described in Corollary \ref{corr:cont_functionalbounds}, the corresponding identified set is an interval (or the empty set). However, Corollary \ref{corr:cont_functionalbounds} characterizes this interval by optimization over the infinite-dimensional spaces $\Pi(\theta)$, which is generally not feasible to compute directly. In this section, we discuss one approach to computing these identified sets by approximating the infinite-dimensional space of densities with a finite-dimensional sieve space and the constraint sets with a finite set of constraints. Similar approximations of identified sets have been used, for example, in \cite{MogstadSantosTorgovitsky2018}. Alternatively, the computational approach developed in \cite{ChristensenConnault2023} could be adapted to our setting. Unlike the sieve-based approach we consider below, the dimension of their optimization problem does not depend on the precision of the density approximation. We leave the application of their approach to our problem to future work.

For simplicity, let $\mathcal{Y}_x = [0,1]$ for $x \in \{0,1\}$. This restriction can be relaxed by linearly transforming the outcome variable so that it has support on the unit interval. We also assume that $\mathcal{F} \coloneqq \mathcal{F}_0 = \mathcal{F}_1$, and therefore $\mathcal{F}_{\text{den}} \coloneqq \mathcal{F}_{\text{den},0} = \mathcal{F}_{\text{den}, 1}$. We also impose assumptions \ref{assn:contsupport} and \ref{assn:compactdensities}. 

We will approximate $\mathcal{F}_{\text{den}}$ by the convex sieve space $\mathcal{F}_M$, defined by 
\begin{align*}
 \mathcal{F}_{M} &\coloneqq \left\{f^M = w^{\top}\textbf{b}^M: w \in \Delta_M \right\},
\end{align*}
where $\textbf{b}^M \coloneqq \{b_0^M, b_1^M, \ldots,b_{M}^M\}$ are the $M$-degree Bernstein basis polynomials scaled by $M + 1$. That is,
$$
b^M_m(y) \coloneqq (M+1) \binom{M}{m} y^m (1-y)^{M-m}
$$
for $m \in \{0, \ldots, M\}$.

Since $\mathcal{F}_M$ is increasing in $M$ and $\bigcup_{M:M > 0}\mathcal{F}_M$ is dense in $\mathcal{F}_{\text{den}}$, $\mathcal{F}_M$ is a sieve space for $\mathcal{F}_{\text{den}}$. We denote the Bernstein polynomial approximation to function $f$ at $y \in [0,1]$ as
\begin{align*}
 (B_M f)(y) &\coloneqq \frac{1}{M + 1}\sum_{m=0}^M f\left(\frac{m}{M}\right) b_m^M(y).
\end{align*}

We also define approximate constraint sets, which are characterized by a finite number of linear equality or inequality constraints. First, we approximate $\mathcal{H}_x$ by the sets,
\begin{align*}
 \mathcal{H}_x^M \coloneqq \{(f_1,\ldots,f_{s_Z}) \in \mathcal{F}_M^{s_Z} : f_j(\cdot) \geq \pi(x \mid z_j) (B_M f_{Y \mid X, Z})( \cdot \mid x, z_j) \text{ for } j = 1, \ldots, s_Z\},
\end{align*}
where $\mathcal{F}_M^{s_Z}$ is the $s_Z$-fold Cartesian product of $\mathcal{F}_M$. In the proposition below, we show that replacing $f_{Y \mid X, Z}$ by its Bernstein approximation is sufficient to characterize this set by a finite number of linear constraints.

Next, we approximate $\mathcal{A}(\theta)$ using a finite set of inequalities. Each model in Section \ref{sec:cont_Y} uses linear inequalities: $\mathcal{A}(\theta) = \{\mathbf{f} \in \mathcal{F}^{s_Z}_{\text{den}} : A(\theta)\mathbf{f} \le a(\theta)\}$. We use a grid of $N$ points in $[0,1]$ (for example, $y_n = n/(N + 1)$ for $n = 1, ..., N$), and then define $\mathcal{A}^{M,N}(\theta)$ as all $\textbf{f} \in \mathcal{F}_M^{s_Z}$ such that $A(\theta) \textbf{f}(y_n) \le a(\theta)$ for each grid point.

The approximate identified set for $\textbf{f}_Y$ is $\Pi_0^{M,N}(\theta) \times \Pi_1^{M,N}(\theta)$, where $\Pi_x^{M,N}(\theta)$ is the intersection of $\mathcal{A}^{M,N}(\theta)$ and $\mathcal{H}_x^M$. The next proposition gives a more convenient representation of this set for computation. Here, $\bar{\Delta}_s^r \coloneqq \left\{\begin{bmatrix} a_1 & \cdots & a_r \end{bmatrix}^{\top} : a_j \in \Delta_s \text{ for } j = 1, \ldots, r\right\}$. $\operatorname{vec}(W)$ is the vectorization of matrix $W$, and $\otimes$ is the Kronecker product.

\begin{proposition}\label{prop:approxIDset}
 For $\mathcal{A}(\theta) = \{\mathbf{f} \in \mathcal{F}_{\text{den}}^{s_Z} : A(\theta) \mathbf{f} \le a(\theta)\}$, $N,M \in \mathbb{N}$, and $\{y_1, \ldots, y_N\} \subset [0,1]$, the approximate constraint sets $\mathcal{A}^{M,N}(\theta)$ and $\mathcal{H}_x^M$ can be represented as
\begin{align*}
 \mathcal{A}^{M,N}(\theta) = \{ W \textbf{b}^M : W \in \mathcal{W}^{M,N}(\theta) \}
\end{align*} 
and
\begin{align*}
 \mathcal{H}_x^M = \{ W \textbf{b}^M : W \in \mathcal{W}_x^M \}
\end{align*}
where
\begin{align*}
 \mathcal{W}^{M,N}(\theta) &\coloneqq \left\{ W \in \bar{\Delta}_{M}^{s_Z} : \left( (B^{M,N})^{\top} \otimes A(\theta) \right) \operatorname{vec}(W) \le \iota_{N} \otimes a(\theta) \right\} \\
 \mathcal{W}_x^{M} &\coloneqq \left\{ \textbf{D}_x \Xi^M_x + \textbf{D}_{1-x} W : W \in \bar{\Delta}_{M}^{s_Z} \right\}.
\end{align*}
In $\mathcal{W}^{M,N}(\theta)$, we define $B^{M,N} \coloneqq \begin{bmatrix} \textbf{b}^M(y_1) & \ldots & \textbf{b}^M(y_{N}) \end{bmatrix}$ and $\iota_{N}$ to be the $N$-dimensional vector of ones. In $\mathcal{W}_x^M$, we define $\textbf{D}_x \coloneqq \operatorname{diag}(\pi(x \mid z_1), \ldots, \pi(x \mid z_{s_Z}))$ and $\Xi^M_x$ to be the $s_Z \times (M + 1)$ matrix with elements $f_{Y \mid X, Z}\left( \frac{m-1}{M} \mid x, z_j\right)$ in the $(j, m)$-th position.
\end{proposition}

This proposition shows that the approximate identified set, $\prod_{x \in \{0, 1\}} \left(\mathcal{A}(\theta)^{M,N} \cap \mathcal{H}_x^M\right)$ can be characterized by a finite number of linear constraints. Following Corollary 2, we use this result to characterize the approximate identified set of a functional of $\textbf{f}_Y$ as the solution to a finite linear program. 

Approximating the functional $\Gamma(\textbf{f}_Y) = \int_{\mathcal{Y}_0} \omega_0(y)^{\top} f_0(y) dy + \int_{\mathcal{Y}_1} \omega_1(y)^{\top} f_1(y) dy$ with a Riemann sum with $L$ points, we can characterize $\underline{\Gamma}^{M,N,L}(\theta)$ as the solution to the linear program,
\begin{align}
\minimize_{\substack{W_1, W_{1,0},\\ W_0, W_{0,1}}} \quad &\frac{1}{L} \sum_{n=0}^{L-1} \left( \omega_1\left(\frac{n}{L}\right)^{\top} W_1 + \omega_0\left(\frac{n}{L}\right)^{\top} W_0 \right) \textbf{b}^M\left(\frac{n}{L}\right) \notag \\
\text{subject to} \quad 
&W_{0}, W_{0,1}, W_{1}, W_{1,0} \in \bar{\Delta}_{M}^{s_Z} \notag \\
&W_1 - \textbf{D}_1 \Xi^M_1 - \textbf{D}_0 W_{1,0} = 0 \label{eq:LPconstraints1}\\
&W_0 - \textbf{D}_0 \Xi^M_0 - \textbf{D}_1 W_{0,1} = 0 \label{eq:LPconstraints2}\\
&\left((B^{M,N})^{\top} \otimes A(\theta)\right) \operatorname{vec}(W_1) \le \iota_{N} \otimes a(\theta) \label{eq:LPconstraints3}\\
&\left((B^{M,N})^{\top} \otimes A(\theta)\right) \operatorname{vec}(W_0) \le \iota_{N} \otimes a(\theta) \label{eq:LPconstraints4}
\end{align} 
The linear inequalities \eqref{eq:LPconstraints3} and \eqref{eq:LPconstraints4} correspond to the constraints that $W_1, W_0 \in \mathcal{W}^{M,N}(\theta)$, and the equality constraints \eqref{eq:LPconstraints1} and \eqref{eq:LPconstraints2} together with the simplex constraints on $W_{x,1-x}$ for $x \in \{0, 1\}$ correspond to the constraints that $W_1 \in \mathcal{W}_1^M$ and $ W_0 \in \mathcal{W}_0^M$ respectively. The optimization program is therefore a linear program in the $s_Z \times (M + 1)$ weight matrices $W_1, W_{1,0}, W_0, W_{0,1}$, which can be solved using standard software.

$\overline{\Gamma}^{M,N,L}(\theta)$ is the solution to the corresponding maximization problem, which is also a linear program.

Since $\Pi^{M,N}(\theta)$ is closed, bounded, and convex, the approximate identified set is 
\begin{align*}
	[\underline{\Gamma}^{M,N,L}(\theta), \overline{\Gamma}^{M,N,L}(\theta)] = \{ \Gamma^{M,N,L}(\textbf{f}_Y) : \textbf{f}_Y \in \Pi^{M,N}(\theta) \}.
\end{align*}

Although we omit a full analysis, we expect that $\underline{\Gamma}^{M,N,L}(\theta)$ and $\overline{\Gamma}^{M,N,L}(\theta)$ will converge to $\underline{\Gamma}(\theta)$ and $\overline{\Gamma}(\theta)$ respectively as $M, N, L \rightarrow \infty$ under suitable regularity conditions.

\section{Empirical Application}\label{subsec:empirical}

Here we revisit the empirical study of peer effects in consumer demand by \cite{GilchristSands2016}. Specifically, they study whether movie viewership is affected by peer viewership choices. They provide evidence that movie viewership can have ``momentum" from one weekend to the next. They argue that this is partly because if a movie does well on its opening weekend, it motivates people to see it in subsequent weekends, so they can discuss it with their peers or attend it as a social event.

Identifying this effect is a challenging empirical problem: an apparent peer effect on consumer demand could simply reflect a common understanding of the movie’s unobserved quality. To address this, the authors use a classic instrumental variables approach, using weather as an instrument for opening weekend viewership. They argue that outdoor activities are a substitute for going to the movies, so days with especially nice weather provide a plausibly negative, exogenous shock to viewership.

While its inherent randomness makes weather an appealing instrument, recent literature has cast doubt on its validity as an instrument in many contexts (e.g., \citealt{Mellon2025}). For this application, we highlight three potential violations of the exclusion assumption: (1)  social learning about movie quality, (2) dynamic consumer behavior, and (3) dynamic behavior by movie studios.

\cite{GilchristSands2016} acknowledge that social learning is an important alternative explanation for the observed momentum in movie viewership. The concern is that consumers may be uncertain about a movie's quality and rely on their peers to learn about it. When viewership is high, there is a higher probability that a consumer has friends who have seen the movie and can share their opinion of it. For more reluctant consumers, they may wait until they have good information about the film's quality before seeing it. This is a similar but distinct mechanism from the social incentive that the authors are interested in.

One approach would be to redefine the ``peer effect'' to include this learning effect; however, \cite{GilchristSands2016} are clear that they are interested in the direct social incentive to see the movie. Instead, they explore whether there are learning effects by testing an implication from a model of social learning in \cite{Young2009}. This auxiliary model introduces several additional strong behavioral and distributional assumptions, and the results are not decisive. They conclude that ``Although our estimates do not rule out some role for learning, taken together the results suggest that the observed momentum is driven in part by a preference for shared experience, and not only by learning.'' \citep[][p.1342]{GilchristSands2016}.

Dynamic behavior could also lead to violations of exclusion. When a consumer skips seeing a particular movie one weekend to enjoy the weather, she may simply plan to see the movie on a future weekend. However, the set of available movies in that future weekend is often different, possibly leading them to make a different choice about what movie to see altogether. Similarly, movie studios may respond to first-weekend viewership by adjusting their advertising strategy, which could affect subsequent viewership.

Finally, we note an additional challenge to the exogeneity condition which \cite{GilchristSands2016} address directly in their main specifications. Movie studios may strategically time movie release dates based on seasonal weather patterns, inducing a correlation between weather shocks and unobserved movie quality. To address this problem, the authors condition on several calendar controls, including the week of the year, the year, and holiday indicators. Since movie studios have to release movies based on their expectations of the weather far in advance rather than short-term forecasts, they argue that this strategic behavior should be captured by these calendar controls. In our analyses, we follow their approach of controlling for these time-of-year variables. However, this could still be insufficient if movie studios use more accurate long-term weather forecasts than the average weather for that week of the year.

These potential violations of the exclusion and exogeneity assumptions motivate the importance of assessing sensitivity in this application.

\subsection{Data and Definitions}\label{data-and-definitions}

We use the dataset assembled by \cite{GilchristSands2016} for our analysis. Viewership data on daily ticket sales is obtained from the Internet Movie Database (IMDb) for all movies released between 2002 and 2013. The sample is restricted to movies that were in theaters for at least six weeks, and uses only data on ticket sales on Friday, Saturday, and Sunday.

The instruments are measures of the weather on each weekend. These data come from Weather Underground and consist of (1) the daily maximum temperature, (2) inches of rain, and (3) inches of snow in \(1,941\) weather stations across the country. In order to create national aggregate measures, weather station-level data is weighted by $\omega_j = \frac{n_j}{\sum_j n_j}$ for each weather station $j$ where $n_j$ is the number of movie theaters for which $j$ is the closest weather station.\footnote{To do this, they first assign each zip code to the closest weather station, and obtain the number of movie theaters in each zip code from the U.S. Census Zip Code Business Patterns data.} For any weather station-level weather measure, $Z_{tj}$, the aggregate instrument is, $Z_{t} = \sum_j \omega_j Z_{tj}$.

We define the potential outcome, $Y_{i}(x)$, to be the viewership of movie $i$ in the second weekend of its release with or without a negative shock to viewership in the opening weekend, $x$. The treatment is binary, with $x = 1$ when opening-weekend viewership is below its 25th percentile. This specification of the treatment is motivated by the observation in \cite{GilchristSands2016} that good weather tends to suppress viewership.

We want to ask whether such a negative shock to initial viewership increases the probability of low viewership in subsequent weekends through peer effects. We begin by defining low viewership in the second weekend analogously to the treatment. Specifically, we consider the summary outcome $\mathbf{1}(Y_i(x) \le \underbar{y})$, where $\underbar{y}$ is the 25th percentile of viewership in the second weekend across all movies. The natural parameter of interest is the average treatment effect (ATE), $\mathbb{E}(\mathbf{1}(Y_i(1) \le \underbar{y}) - \mathbf{1}(Y_i(0) \le \underbar{y}))$. This is the effect of a negative shock to opening weekend viewership on the probability of low viewership in the second weekend. Moving beyond this coarse measure of low viewership in the second weekend, we also consider quantile treatment effects across the distribution of viewership in that weekend. That is for, a range of quantiles $\tau \in (0,1)$, we consider the parameter $\text{QTE}(\tau) = Q_{Y_i(1)}(\tau) - Q_{Y_i(0)}(\tau)$ where $Q_{Y_i(x)}(\tau)$ is the $\tau$th quantile of the distribution of $Y_i(x)$.

To minimize endogeneity between movie quality and opening weekend weather, we follow the approach of Gilchrist and Sands (2016) and residualize all variables (viewership in the first and second weekends and the weather instrument) using a set of week-of-year dummies. We use their preferred weather instrument, the share of theaters with a daily high temperature between 75 and 80 degrees Fahrenheit, which we discretize into quintiles. Finally, we condition on this same weather variable on the second weekend. This helps control for potential serial correlation in weather across weekends, which is not captured by the week-of-year dummies.

\subsection{Sensitivity Analysis}\label{sensitivity-analysis}

We begin with the discretized outcome. Under the baseline of exogeneity, we find that a negative shock on viewership in the initial weekend increases the probability of low viewership in the second weekend. The estimated identified set for the ATE is $[0.04, 0.87]$. This result, which bounds the ATE above zero, is qualitatively consistent with the conclusion of Gilchrist and Sands (2016), who find a positive effect of opening weekend viewership on subsequent weekend viewership using a 2SLS estimator. While the lower bound is small, this means that peer effects increase the probability of low viewership in the second weekend by at least $4.4\%$, which is a quantitatively important effect size.

We find, however, that this conclusion is sensitive to relatively small violations of the exogeneity assumption. In Table \ref{tab:ate_bounds} we present the estimated ATE bounds for different levels of $c$-dependence. The interval between the lower and upper lines is the identified set for the ATE at each level of $c$-dependence. Even at low levels of $c$-dependence, the identified set for the ATE includes zero. The lowest level of $c$-dependence at which the identified set for the ATE includes $0$ -- the \textit{breakdown point} -- is $0.015$, or when the latent propensity score is allowed to be 1.5 percentage points away from the observed propensity score.

\begin{table}[tbh]
\begin{center}
\begin{tabular}{cc}
\hline
$c$ & Estimated ATE bounds \\
\hline
0.000 & [0.038, 0.872] \\
0.010 & [0.012, 0.880] \\
0.015 & [0.000, 0.883] \\
0.025 & [-0.024, 0.889] \\
0.050 & [-0.055, 0.901] \\
0.100 & [-0.071, 0.917] \\
0.200 & [-0.071, 0.928] \\
0.500 & [-0.071, 0.929] \\
1.000 & [-0.071, 0.929] \\
\hline
\end{tabular}

\caption{Estimated ATE bounds for different values of sensitivity parameter $c$.}
\label{tab:ate_bounds}
\end{center}
\end{table}

To explore the distributional effects of a negative shock to opening weekend viewership, we now turn to the quantile treatment effects (QTE) for the continuous outcome $Y_i(x)$. In Table \ref{tab:qte_bounds}, we report the identified set of the QTE across several quantiles and different levels of $c$-dependence. Consistent with the results of the discretized outcome, we find that under the baseline assumption of exogeneity, the identified set for the QTE at the $10$th and $25$th percentile is negative and bounded away from zero. A negative shock to opening weekend viewership causes the 25th percentile of viewership in the second weekend to decrease by at least $0.39$ million tickets. These results, however, hold only for the bottom half of the distribution of potential outcomes. At the $50$th, $75$th, and $90$th percentiles, the identified set is very wide and includes zero.

\begin{table}[tbh]
\begin{center}
\begin{tabular}{lccccc}
\toprule
Percentile           & 10\% & 25\% & 50\% & 75\% & 90\% \\
\midrule
$c$ = 0.00   & [-2.94, -0.60] & [-3.45, -0.39] & [-4.02, 8.89] & [-3.11, 7.99] & [-5.57, 6.39] \\
$c$ = 0.02   & [-2.97, 1.90] & [-3.50, -0.13] & [-4.15, 9.02] & [-3.90, 8.11] & [-9.92, 6.57] \\
$c$ = 0.10   & [-3.06, 2.29] & [-3.58, 0.90] & [-4.29, 9.15] & [-7.16, 8.32] & [-10.05, 7.08] \\
\bottomrule
\end{tabular}

\caption{QTE bounds under $c$-dependence. This table shows the estimated identified set for $Q_{Y_i(1)}(\tau) - Q_{Y_i(0)}(\tau)$ for a range of values of $\tau$ (columns) and levels of $c$-dependence (rows).}
\label{tab:qte_bounds}
\end{center}
\end{table}

To see why the identified set for the QTE is much less informative for higher quantiles, it is useful to examine the identified sets for the potential outcome CDFs directly. Figure \ref{fig:cdf_bounds} shows the upper and lower bounds on the CDF for $(Y_i(1), Y_i(0))$ at different levels of $c$-dependence. The first panel shows the bounds under exogeneity, while the second shows a $c$-dependence level of $0.1$. There is an asymmetry in the bounds of the distributions of potential outcomes, with much tighter bounds for the potential outcome with $x = 0$ in which viewership in the opening weekend is above the $25$th percentile. This is because there is a much larger mass of observations with $X = 0$ than with $X = 1$. In addition, the data is largely uninformative about the top half of the distribution of $Y_i(1)$. This reflects the fact that nearly all of the observed mass of $Y$ conditional on $X = 1$ is in the lower half of the support of $Y$. Since we make no monotonicity assumption or other shape restriction, the bounds on the CDF of $Y_i(1)$ have no other restriction except for the lower bound from the mass below $0$. 

\begin{figure}[tbh]
\centering
\input{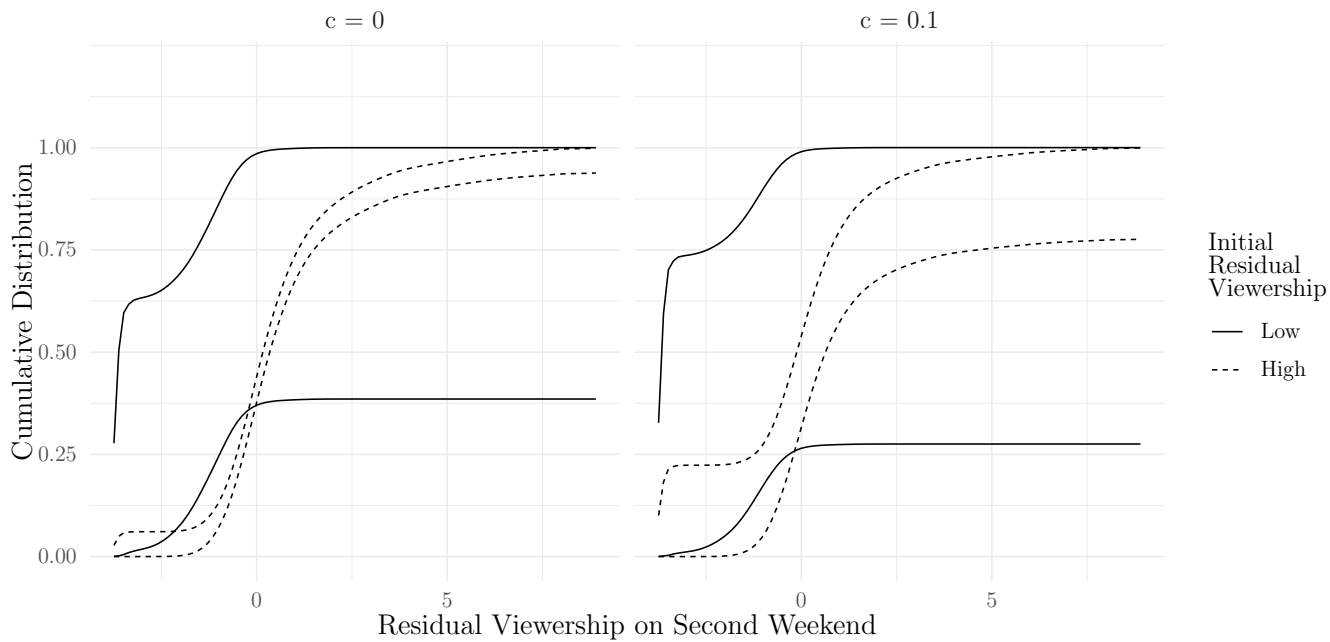}
\caption{Distributional bounds. These plots show the bounds on the cumulative distribution function of $Y_i(x)$ at two values of $c$-dependence.}
\label{fig:cdf_bounds}
\end{figure}

\section{Conclusion}

We introduced a new, computationally tractable approach for conducting sensitivity to the instrument exclusion and exogeneity assumptions. Our approach does not impose any kind of monotonicity assumption in the first stage, and allows for arbitrarily heterogeneous treatment effects. We did this by developing a unifying sensitivity model which nests several well known approaches to continuously parameterizing relaxations of statistical independence assumptions from the literature. We showed that, under those relaxations, identified sets for parameters like ATE and QTE are solutions to linear programs. Our approach can be used when the outcome is discrete or continuous, and when there are one or multiple discretely supported instruments.

We illustrated the practical value of our results in an empirical study of peer effects in movie viewership. There our sensitivity analysis shows that although ATE is positive under full exclusion and exogeneity (meaning peer effects are present), that conclusion is highly sensitive to minor relaxations of the exclusion and exogeneity assumptions. Overall, our results allow researchers to transparently study and report the robustness of their instrumental variable conclusions to violations of exclusion or exogeneity.

\bibliographystyle{econometrica}
\bibliography{FF_paper}

\appendix
\newpage
\section{Proofs for Section \ref{sec:binary_outcome}: Binary Outcomes}\label{appsec:binary_outcomes_proofs}

\begin{proof}[Proof of Proposition \ref{prop:noassn_discrete}]
 We have that 
\begin{align*}
 \Prob(Y(x) = 1|Z=z) &= \Prob(Y(x) = 1,X=x|Z=z) + \Prob(Y(x) = 1,X = 1-x\mid Z=z)\\
 &= \Prob(Y = 1,X=x\mid Z=z) + \Prob(Y(x) = 1\mid X = 1-x,Z=z)\pi(1-x\mid z)\\
 &\in [\Prob(Y=1,X=x\mid Z=z),\Prob(Y=1,X=x\mid Z=z) + \pi(1-x\mid z)]
\end{align*}
by $\Prob(Y(x) = 1\mid X = 1-x,Z=z) \in [0,1]$. All these conditional probabilities are well defined by Assumption \ref{assn:NonTrivialinstrument}. This inclusion holds for all $x,z \in \{0,1\}$, therefore $\textbf{p}_Y \in \mathcal{H}_0 \times \mathcal{H}_1$. 

To show sharpness, let $(p_{00},p_{01},p_{10},p_{11}) = \textbf{p} \in \mathcal{H}_0 \times \mathcal{H}_1$. We will find a distribution for $(Y(0),Y(1))\mid X,Z$ that is consistent with $\textbf{p}_Y = \textbf{p}$ and the known distribution of $Y\mid X,Z$. For $x,z \in \{0,1\}$, let
\begin{align*}
 \Prob(Y(0) = y_0,Y(1) = y_1\mid X=x,Z=z) &= \Prob(Y(0) = y_0\mid X=x,Z=z) \cdot \Prob(Y(1) = y_1\mid X=x,Z=z)\\
 \Prob(Y(x) = 1\mid X= x,Z=z) &= \Prob(Y=1\mid X=x,Z=z)\\
 \Prob(Y(x) = 1\mid X=1-x,Z=z) &= \frac{p_{xz} - \Prob(Y=1\mid X=x,Z=z)\pi(x \mid z)}{\pi(1-x\mid z)}.
\end{align*}
By $\textbf{p} \in \mathcal{H}_0 \times \mathcal{H}_1$, $(p_{xz} - \Prob(Y=1\mid X=x,Z=z)\pi(x \mid z))/\pi(1-x\mid z) \in [0,1]$ and hence $\Prob(Y(x) = y_x\mid X=1-x,Z=z)$ is a probability. This choice of $\Prob(Y(0) = y_0,Y(1) = y_1\mid X=x,Z=z)$ has a distribution of $Y\mid X,Z$ that coincides with its known distribution. Finally, we can compute
\begin{align*}
 \Prob(Y(x) = 1\mid Z=z) &= \Prob(Y = 1,X=x\mid Z=z) + \Prob(Y(x) = 1\mid X = 1-x,Z=z)\pi(1-x\mid z)\\
 &= \Prob(Y = 1,X=x\mid Z=z) + \frac{p_{xz} - \Prob(Y=1\mid X=x,Z=z)\pi(x \mid z)}{\pi(1-x\mid z)}\pi(1-x\mid z)\\
 &= p_{xz}.
\end{align*}
Therefore, $\mathcal{H}_0 \times \mathcal{H}_1$ is sharp.
\end{proof}

\begin{proof}[Proof of Lemma \ref{lemma:equivalence}]
 By Assumption \ref{assn:exog_or_excl}, we consider the lemma's result under exogeneity and under weak exclusion  separately.
 
 First, suppose exogeneity of $Z$ holds. In this case, if $p_Y(x,1) = p_Y(x,0)$, then $\P(Y(x,1) = 1) = \P(Y(x,0) = 1)$. By exogeneity of $Z$, this equivalently implies $\P(Y(x,1) = 1 \mid Z) = \P(Y(x,0) = 1 \mid Z)$ almost surely, showing that weak exclusion also holds. For the reverse direction, note that weak exclusion immediately implies $p_Y(x,1) = p_Y(x,0)$, so this implication is direct and omitted.
 
 Now consider the case where weak exclusion of $Z$ holds. If $p_Y(x,1) = p_Y(x,0)$, then we have $\P(Y(x,1) = 1 \mid  Z = 1)  = \P(Y(x,0) = 1 \mid  Z = 0)$. By weak exclusion, $\P(Y(x,0) = 1 \mid  Z = 1) = \P(Y(x,1) = 1 \mid  Z = 1)  = \P(Y(x,0) = 1 \mid  Z = 0) = \P(Y(x,1) = 1 \mid  Z = 0)$. Therefore, it follows that $Z$ is independent of both $Y(x,1)$ and $Y(x,0)$. The reverse implication that exogeneity implies $p_Y(x,1) = p_Y(x,0)$ is again immediate and omitted.
\end{proof}

\begin{proof}[Proof of Proposition \ref{prop:high-level_relax_discrete}]
 This proposition follows from lemmas \ref{lemma:MSM_prim_assn}--\ref{lemma:KS_prim_assn}, which verify that the four conditions of Assumption \ref{assn:ind_relax} hold for the corresponding sensitivity model.
\end{proof}

\begin{lemma}\label{lemma:MSM_prim_assn}
Let Assumption \ref{assn:NonTrivialinstrument} hold. The correspondence defined in equation \eqref{eq:MSM_discrete_correspondence} satisfies Assumption \ref{assn:ind_relax}.
\end{lemma}

\begin{proof}[Proof of Lemma \ref{lemma:MSM_prim_assn}]

First, define $\mathcal{A}_\text{MSM}(\theta;x) \coloneqq \{\textbf{p} \in [0,1]^2: A_\text{MSM}(\lambda)\textbf{p} \leq a_\text{MSM}(\lambda)\}$. 

\textbf{Part 1:} From the definitions of $(A_\text{MSM}(\lambda),a_\text{MSM}(\lambda))$, we can directly see that
\begin{align*}
 \mathcal{A}_\text{MSM}(0;x) &= \{(p_0,p_1) \in [0,1]^2: p_{0} - p_{1} \leq 0, p_{1} - p_{0} \leq 0\} = \{(p_0,p_1) \in [0,1]^2: p_{0} = p_{1}\}
\end{align*}
and that
\begin{align*}
 \mathcal{A}_\text{MSM}(1;x) &= \{(p_0,p_1) \in [0,1]^2: - p_{1} \leq 0,  - p_{0} \leq 0, p_{0} \leq 1, p_{1} \leq 1\}= [0,1]^2.
\end{align*}

\textbf{Part 2:} Let $\lambda' \geq \lambda$ and suppose $\textbf{p}  \in \mathcal{A}_\text{MSM}(\lambda;x)$. Therefore,  $A_\text{MSM}(\lambda)\textbf{p} - a_\text{MSM}\leq 0$. This implies that 
\begin{align*}
 A_\text{MSM}(\lambda')\textbf{p} - a_\text{MSM}(\lambda') &= \begin{pmatrix}
 (1 - \lambda')p_{0} - p_{1}\\
 -p_{0} + (1 - \lambda')p_{1}\\
 -(1-p_{0})\lambda' + p_{1} - p_{0}\\
 -(1-p_{1})\lambda' + p_{0} - p_{1}
 \end{pmatrix}\leq \begin{pmatrix}
 (1 - \lambda)p_{0} - p_{1}\\
 -p_{0} + (1 - \lambda)p_{1}\\
 -(1-p_{0})\lambda + p_{1} - p_{0}\\
 -(1-p_{1})\lambda + p_{0} - p_{1}
 \end{pmatrix}\leq 0
\end{align*}
from $\lambda \leq \lambda' \leq 1$ and from $p_{0},p_{1} \in [0,1]$. Therefore, $\textbf{p} \in \mathcal{A}_\text{MSM}(\lambda';x)$.

\textbf{Part 3:} $\mathcal{A}_\text{MSM}(\lambda;x)$ trivially defines a bounded set defined as the intersection of finitely many closed half-planes. Hence, it is a closed and convex polytope. Since $\mathcal{A}_\text{MSM}(\lambda)$ is a Cartesian product of closed convex polytopes, it is also a closed convex polytope.

\textbf{Part 4:} We break this part by first showing that the correspondence $\mathcal{A}_x:[0,1]\rightrightarrows [0,1]^2$ is upper-hemicontinuous (uhc), and then second showing that it is lower-hemicontinuous (lhc).

To show uhc, note that $\mathcal{A}_\text{MSM}(\cdot;x)$ is compact-valued since $\mathcal{A}_\text{MSM}(\lambda;x)$ is a closed and bounded set for all $\lambda \in [0,1]$. Let $\lambda_n \to \lambda$, $\textbf{p}_n \in \mathcal{A}_\text{MSM}(\lambda_n;x)$ and $\textbf{p}_n \to \textbf{p} \in [0,1]^2$ as $n\to\infty$. We can see that $\textbf{p} \in \mathcal{A}_\text{MSM}(\lambda;x)$ because
\begin{align*}
 A_\text{MSM}(\lambda)\textbf{p} - a_\text{MSM}(\lambda) &= \lim_{n \to \infty} \left(A_\text{MSM}(\lambda_n)\textbf{p}_n - a_\text{MSM}(\lambda_n)\right)\leq  \lim_{n \to \infty} (0) = 0,
\end{align*}
where the equality follows from the continuity of $A_\text{MSM}(\lambda)\textbf{p} - a_\text{MSM}(\lambda)$ in $(\lambda,\textbf{p})$. Thus, $\mathcal{A}_\text{MSM}(\lambda;x)$ is uhc.

To show lhc, let $\lambda_n \to \lambda$ and fix $\textbf{p} =(p_0,p_1) \in \mathcal{A}_\text{MSM}(\lambda;x)$. $\mathcal{A}_\text{MSM}(\lambda;x)$ is lhc if we can find $\textbf{p}_n \in \mathcal{A}_\text{MSM}(\lambda_n;x)$ such that $\textbf{p}_n \to \textbf{p}$. 

Let 
\begin{align*}
 \varepsilon_{1,n} &\coloneqq \frac{(1-\lambda_n)p_0 - p_1}{(1-\lambda_n)p_0 - p_1 + \lambda_n (p_0 + p_1)/2}\indicator((1-\lambda_n)p_0 - p_1 > 0).
\end{align*}
 Note that the denominator is nonzero when $(1-\lambda_n)p_0 - p_1 >0$ because $\lambda_n (p_0+p_1)/2 \geq 0$. Therefore, $\varepsilon_{1,n} \in [0,1]$. For $\varepsilon \in [0,1]$, define
\begin{align*}
 \textbf{p}(\varepsilon) &\coloneqq (1-\varepsilon)\textbf{p} + \varepsilon \begin{pmatrix}
  (p_0+p_1)/2\\
  (p_0+p_1)/2
 \end{pmatrix}.
\end{align*}
We can show that $\textbf{p}(\varepsilon_{1,n})$ satisfies the first inequality characterizing $\mathcal{A}_\text{MSM}(\lambda_n;x)$. To see this,
\begin{align*}
 (1-\lambda_n)p_0(\varepsilon_{1,n}) - p_1(\varepsilon_{1,n}) &= (1-\lambda_n)\left((1-\varepsilon_{1,n})p_0 + \varepsilon_{1,n}(p_0+p_1)/2\right) - \left((1-\varepsilon_{1,n})p_1 + \varepsilon_{1,n}(p_0+p_1)/2\right)\\
 &= (1-\lambda_n)p_0 - p_1 - \varepsilon_{1,n}((1-\lambda_n)p_0 - p_1 +\lambda_n(p_0 + p_1)/2)\\
 &= (1-\lambda_n)p_0 - p_1 - ((1-\lambda_n)p_0 - p_1)\indicator((1-\lambda_n)p_0 - p_1 > 0)\\
 &= ((1-\lambda_n)p_0 - p_1)\indicator((1-\lambda_n)p_0 - p_1 \leq 0)\\
 &\leq 0.
\end{align*}
We can also show that $\textbf{p}(\varepsilon')$ satisfies the first inequality characterizing $\mathcal{A}_\text{MSM}(\lambda_n;x)$ for all $\varepsilon' \in [\varepsilon_{1,n},1]$. To see this,
\begin{align*}
 &(1-\lambda_n)p_0(\varepsilon') - p_1(\varepsilon')\\
 &= (1-\lambda_n)p_0 - p_1 - \varepsilon'((1-\lambda_n)p_0 - p_1 + \lambda_n(p_0 + p_1)/2)\\
 &= \left((1-\lambda_n)p_0 - p_1 - \varepsilon'((1-\lambda_n)p_0 - p_1 + \lambda_n(p_0 + p_1)/2)\right)1(p_1 \geq p_0)\\
 &+ \left((1-\lambda_n)p_0(\varepsilon_{1,n}) - p_1(\varepsilon_{1,n}) + (\varepsilon_{1,n} - \varepsilon')((1-\lambda_n)p_0 - p_1 + \lambda_n(p_0 + p_1)/2)\right)\indicator(p_1 < p_0)\\
 &= \Big(\underbrace{(p_0 - p_1)(1 - \varepsilon'(1-\lambda_n/2))}_{\leq 0} - \lambda_n p_0\Big)\indicator(p_1 \geq p_0)\\
 &+ \left(\underbrace{(1-\lambda_n)p_0(\varepsilon_{1,n}) - p_1(\varepsilon_{1,n})}_{\leq 0} + \underbrace{(\varepsilon_{1,n} - \varepsilon')}_{\leq 0}\underbrace{(1-\lambda_n/2)(p_0 - p_1)}_{\geq 0}\right)\indicator(p_1 < p_0)\\
 & \leq 0.
\end{align*}

Finally, we can see that $\varepsilon_{1,n} \to 0$ as $n\to \infty$ because 
\begin{align*}
 ((1-\lambda_n)p_0 - p_1)\indicator((1-\lambda_n)p_0 - p_1 > 0) &= \max\{(1-\lambda_n)p_0 - p_1,0\} \to \max\{(1-\lambda)p_0 - p_1,0\} = 0
\end{align*}
because $(1-\lambda)p_0 - p_1 \leq 0$, which follows from $(p_0,p_1) \in \mathcal{A}_\text{MSM}(\lambda;x)$.

Define
\begin{align*}
 \varepsilon_{2,n} &\coloneqq \frac{(1-\lambda_n)p_1 - p_0}{(1-\lambda_n)p_1 - p_0 + \lambda_n (p_0 + p_1)/2}\indicator((1-\lambda_n)p_1 - p_0 > 0)\\
 \varepsilon_{3,n} &\coloneqq \frac{p_1 - (1-\lambda_n)p_0 - \lambda_n}{(1 - \lambda_n/2)(p_1-p_0)}\indicator(p_1 - (1-\lambda_n)p_0 - \lambda_n > 0)\\
 \varepsilon_{4,n} &\coloneqq \frac{p_0 - (1-\lambda_n)p_1 - \lambda_n}{(1 - \lambda_n/2)(p_0-p_1)}\indicator(p_0 - (1-\lambda_n)p_1 - \lambda_n > 0).
\end{align*}
As we did above for $\varepsilon_{1,n}$, we can verify that $\textbf{p}(\varepsilon_{i,n})$ satisfies the $i$th equation in $\mathcal{A}_\text{MSM}(\lambda_n;x)$, $\textbf{p}(\varepsilon')$ satisfies the $i$th equation for all $\varepsilon' \in [\varepsilon_{i,n},1]$, and that $\varepsilon_{i,n} \to 0$ as $n\to\infty$ for $i \in \{2,3,4\}$. Let $\varepsilon_n \coloneqq \max_{i \in \{1,2,3,4\}} \varepsilon_{i,n}$. Since $\varepsilon_n \geq \varepsilon_{i,n}$ for all $i$, $\textbf{p}(\varepsilon_n)\in \mathcal{A}_\text{MSM}(\lambda_n;x)$. Moreover, $\lim_{n\to\infty} \textbf{p}(\varepsilon_n) = \textbf{p}(\lim_{n\to\infty}\max_{i \in \{1,2,3,4\}} \varepsilon_{i,n}) =  \textbf{p}(0) = \textbf{p}$. Therefore, $\mathcal{A}_\text{MSM}(\lambda;x)$ is lhc. Since it is also uhc, it is continuous. The Cartesian product of continuous compact-valued correspondences is continuous by Theorem 11.25 in \cite{Border1985}.
\end{proof}

\begin{lemma}\label{lemma:cdep_prim_assn}
Let Assumption \ref{assn:NonTrivialinstrument} hold. The correspondence defined in equation \eqref{eq:cdep_discrete_correspondence} satisfies Assumption \ref{assn:ind_relax}.
\end{lemma}

\begin{proof}[Proof of Lemma \ref{lemma:cdep_prim_assn}]

First, let $\mathcal{A}_\text{$c$-dep}(\theta;x) \coloneqq \{\textbf{p} \in [0,1]^2: A_\text{$c$-dep}(c)\textbf{p} \leq a_\text{$c$-dep}(c)\}$.

\textbf{Part 1:} From the definitions of $(A_\text{$c$-dep}(c),a_\text{$c$-dep}(c))$ and from $k_z(0) = 1$, we can directly see that
\begin{align*}
 \mathcal{A}_\text{$c$-dep}(0;x) &= \{\textbf{p} \in [0,1]^2: p_{0} - p_{1} \leq 0, p_{1} - p_{0} \leq 0\} = \{\textbf{p} \in [0,1]^2: p_{0} = p_{1}\}.
\end{align*}
From $k_z(1) = 0$, we can also see that
\begin{align*}
 \mathcal{A}_\text{$c$-dep}(1;x) &= \{\textbf{p} \in [0,1]^2: -p_{1} \leq 0,-p_{0} \leq 0, p_{1} \leq 1,  p_{1} \leq 1\}= [0,1]^2.
\end{align*}

\textbf{Part 2:} Let $c' \geq c$ and suppose $\textbf{p} = (p_0,p_1) \in \mathcal{A}_\text{$c$-dep}(c;x)$. Therefore, $A_\text{$c$-dep}(c)\textbf{p} - a_\text{$c$-dep}(c) \leq 0$. Then,
\begin{align*}
 A_\text{$c$-dep}(c')\textbf{p} - a_\text{$c$-dep}(c') &= \begin{pmatrix}
 k_0(c')p_{0} - p_{1} \\
 -p_{0} + k_1(c') p_{1}\\
 -k_0(c')p_{0} + p_{1} - 1 + k_0(c')\\
 p_{0} - k_1(c')p_{1} - 1 + k_1(c')
 \end{pmatrix}\\
 &= \begin{pmatrix}
 k_0(c')p_{0} - p_{1} \\
 -p_{0} + k_1(c') p_{1}\\
 p_{1} - 1 + k_0(c')(1-p_{0})\\
 p_{0}  - 1 + k_1(c')(1-p_{1})
 \end{pmatrix}\\
 &\leq \begin{pmatrix}
 k_0(c)p_{0} - p_{1} \\
 -p_{0} + k_1(c) p_{1}\\
 p_{1} - 1 + k_0(c)(1-p_{0})\\
 p_{0}  - 1 + k_1(c)(1-p_{1})
 \end{pmatrix}\\
 &\leq 0.
\end{align*}
The second-to-last inequality follows from $k_z(c)$ being nonincreasing for $c\in[0,1]$. Therefore, $\textbf{p} \in \mathcal{A}_\text{$c$-dep}(c';x)$.

\textbf{Part 3:} Similar to the proof of part 3 for Lemma \ref{lemma:MSM_prim_assn}.

\textbf{Part 4:} We again break this part by first showing that the correspondence $\mathcal{A}_\text{$c$-dep}(\cdot;x):[0,1]\rightrightarrows [0,1]^2$ is uhc and lhc.

To show uhc, note that $\mathcal{A}_\text{$c$-dep}(\cdot,x)$ is compact-valued since $\mathcal{A}_\text{$c$-dep}(c;x)$ is a closed and bounded set for all $c \in [0,1]$. Let $c_n \to c$, $\textbf{p}_n \in \mathcal{A}_\text{$c$-dep}(c_n;x)$ and $\textbf{p}_n \to \textbf{p}$ as $n\to\infty$. We can see that $\textbf{p} \in \mathcal{A}_\text{$c$-dep}(c;x)$ because
\begin{align*}
 A_\text{$c$-dep}(c)\textbf{p} - a_\text{$c$-dep}(c) &= \lim_{n \to \infty} \left(A_\text{$c$-dep}(c_n)\textbf{p}_n - a_\text{$c$-dep}(c_n)\right)\leq \lim_{n\to\infty}(0) =  0,
\end{align*}
which follows from the continuity of $A_\text{$c$-dep}(c)\textbf{p} - a_\text{$c$-dep}(c)$ in $(c,\textbf{p})$, which itself follows from the continuity of $k_z(c)$ in $c$. Therefore $\mathcal{A}_\text{$c$-dep}(c;x)$ is uhc.

To show lhc, let $c_n \to c$ and fix $\textbf{p} \in \mathcal{A}_\text{$c$-dep}(c;x)$. $\mathcal{A}_\text{$c$-dep}(c;x)$ is lhc if we can find $\textbf{p}_n \in \mathcal{A}_\text{$c$-dep}(c_n;x)$ such that $\textbf{p}_n \to \textbf{p}$ as $n\to\infty$. 

If $c = 0$, let $\textbf{p}_n = \textbf{p}$, which is in $\mathcal{A}_\text{$c$-dep}(c_n;x)$ by $\mathcal{A}_\text{$c$-dep}(0;x) \subseteq \mathcal{A}_\text{$c$-dep}(c_n;x)$.

When $c> 0$, we construct $\textbf{p}_n$ as in the proof of Lemma \ref{lemma:MSM_prim_assn}. Let
\begin{align*}
 \varepsilon_{1,n} &\coloneqq \frac{k_0(c_n)p_0 - p_1}{k_0(c_n)p_0 - p_1 + (1-k_0(c_n))(p_0 + p_1)/2}\indicator(k_0(c_n)p_0 - p_1 > 0)
\end{align*}

Note that the denominator is nonzero when $k_0(c_n)p_0 - p_1 >0$ because $(1-k_0(c_n))(p_0 + p_1)/2 \geq 0$. Therefore, $\varepsilon_{1,n} \in [0,1]$. For $\varepsilon \in [0,1]$ define 
\begin{align*}
 \textbf{p}(\varepsilon) &\coloneqq (1-\varepsilon)\textbf{p} + \varepsilon \begin{pmatrix}
  (p_0+p_1)/2\\
  (p_0+p_1)/2
 \end{pmatrix}.
\end{align*}
We can show that $\textbf{p}(\varepsilon_{1,n})$ satisfies the first inequality characterizing $\mathcal{A}_\text{$c$-dep}(c_n,x)$ since
\begin{align*}
 k_0(c_n)p_0(\varepsilon_{1,n}) - p_1(\varepsilon_{1,n}) &= k_0(c_n)\left((1-\epsilon_n)p_0 + \varepsilon_{1,n}(p_0+p_1)/2\right) - \left((1-\varepsilon_{1,n})p_1 + \varepsilon_{1,n}(p_0+p_1)/2\right)\\
 &= k_0(c_n)p_0 - p_1 - \varepsilon_{1,n}(k_0(c_n)p_0 - p_1 +(1-k_0(c_n))(p_0 + p_1)/2)\\
 &= k_0(c_n)p_0 - p_1 - (k_0(c_n)p_0 - p_1)\indicator(k_0(c_n)p_0 - p_1 > 0)\\
 &= (k_0(c_n)p_0 - p_1)\indicator(k_0(c_n)p_0 - p_1 \leq 0)\\
 &\leq 0.
\end{align*}
We can show that $\textbf{p}(\varepsilon')$ satisfies the first equality for all $\varepsilon' \in [\varepsilon_{1,n},1]$. To see this,
\begin{align*}
 &k_0(c_n)p_0(\varepsilon') - p_1(\varepsilon')\\
 &= k_0(c_n)p_0 - p_1 - \varepsilon'(k_0(c_n)p_0 - p_1 + (1-k_0(c_n))(p_0 + p_1)/2)\\
 &= \left(k_0(c_n)p_0 - p_1 - \varepsilon'(k_0(c_n)p_0 - p_1 + (1- k_0(c_n))(p_0 + p_1)/2)\right)\indicator(p_1 \geq p_0)\\
 &+ \left(k_0(c_n)p_0(\varepsilon_{1,n}) - p_1(\varepsilon_{1,n}) + (\varepsilon_{1,n} - \varepsilon')(k_0(c_n)p_0 - p_1 + (1-k_0(c_n))(p_0 + p_1)/2)\right)\indicator(p_1 < p_0)\\
 &= \left(\underbrace{(p_0 - p_1)(1 - \varepsilon'(1-(1-k_0(c_n))/2))}_{\leq 0} - (1-k_0(c_n)) p_0)\right)\indicator(p_1 \geq p_0)\\
 &+ \left(\underbrace{k_0(c_n)p_0(\varepsilon_{1,n}) - p_1(\varepsilon_{1,n})}_{\leq 0} + \underbrace{(\varepsilon_{1,n} - \varepsilon')}_{\leq 0}\underbrace{(1-(1-k_0(c_n))/2)(p_0 - p_1)}_{\geq 0}\right)\indicator(p_1 < p_0)\\
 & \leq 0.
\end{align*}

Finally, we can see that $\varepsilon_{1,n} \to 0$ as $n\to \infty$ because 
\begin{align*}
 (k_0(c_n)p_0 - p_1)\indicator(k_0(c_n)p_0 - p_1 > 0) &= \max\{k_0(c_n)p_0 - p_1,0\}\\
 &\to \max\{k_0(c)p_0 - p_1,0\} = 0.
\end{align*}
The limit follows from the continuity of $k_0$ and the maximum, and the last equality follows from $(p_0,p_1) \in \mathcal{A}_\text{$c$-dep}(c;x)$.

Define
\begin{align*}
 \varepsilon_{2,n} &\coloneqq \frac{k_1(c_n)p_1 - p_0}{k_1(c_n)p_1 - p_0 + (1-k_1(c_n)) (p_0 + p_1)/2}\indicator(k_1(c_n)p_1 - p_0 > 0)\\
 \varepsilon_{3,n} &\coloneqq \frac{p_1 - k_0(c_n)p_0 - (1-k_0(c_n))}{(1 - (1-k_0(c_n))/2)(p_1-p_0)}\indicator(p_1 - k_0(c_n)p_0 - (1-k_0(c_n)) > 0)\\
 \varepsilon_{4,n} &\coloneqq \frac{p_0 - k_1(c_n)p_1 - (1-k_1(c_n))}{(1 - (1-k_1(c_n))/2)(p_0-p_1)}\indicator(p_0 - k_1(c_n)p_1 - (1-k_1(c_n)) > 0).
\end{align*}
As we did above for $\varepsilon_{1,n}$, we can verify that $\textbf{p}(\varepsilon_{i,n})$ satisfies the $i$th equation in $\mathcal{A}_\text{$c$-dep}(c_n;x)$, $\textbf{p}(\varepsilon')$ satisfies the $i$th equation for all $\varepsilon' \in [\varepsilon_{i,n},1]$, and that $\varepsilon_{i,n} \to 0$ as $n\to\infty$ for $i \in \{2,3,4\}$. Let $\varepsilon_n = \max_{i \in \{1,2,3,4\}} \varepsilon_{i,n}$. Then, $\textbf{p}(\varepsilon_n)\in \mathcal{A}_\text{$c$-dep}(c_n;x)$ and $\lim_{n\to\infty} \textbf{p}(\varepsilon_n) = \textbf{p}(\lim_{n\to\infty}\max_{i \in \{1,2,3,4\}} \varepsilon_{i,n}) = \textbf{p}(0) = \textbf{p}$. Therefore, $\mathcal{A}_\text{$c$-dep}(c;x)$ is lhc. Since it is also uhc, it is continuous. We conclude that $\mathcal{A}_\text{$c$-dep}(c)$ is continuous.
\end{proof}

\begin{lemma}\label{lemma:KS_prim_assn}
Let Assumption \ref{assn:NonTrivialinstrument} hold. The correspondence defined in \eqref{eq:KS_discrete_correspondence} satisfies Assumption \ref{assn:ind_relax}.
\end{lemma}

\begin{proof}[Proof of Lemma \ref{lemma:KS_prim_assn}]
First, define $\mathcal{A}_\text{KS}(\theta;x) \coloneqq \{\textbf{p} \in [0,1]^2 : A_\text{KS}\textbf{p} \leq a_\text{KS}(K)\}$. We show that the four components of Assumption \ref{assn:ind_relax} hold. 

\textbf{Part 1:} From the definitions of $(A_\text{KS},a_\text{KS}(K))$, we can directly see that
\begin{align*}
 \mathcal{A}_\text{KS}(0;x) &= \{(p_0,p_1) \in [0,1]^2: p_{0} - p_{1} \leq 0, p_{1} - p_{0} \leq 0\} = \{\textbf{p}_x \in [0,1]^2: p_{0} = p_{1}\}
\end{align*}
and that
\begin{align*}
 \mathcal{A}_\text{KS}(1;x) &= \{(p_0,p_1) \in [0,1]^2: -1 \leq p_{0} - p_{1} \leq 1\}= [0,1]^2.
\end{align*}

\textbf{Part 2:} This follows from $a_\text{KS}(K) \leq a_\text{KS}(K')$ when $K \leq K'$.

\textbf{Part 3:} $\mathcal{A}_\text{KS}(K;x)$ is a bounded set defined as the intersection of finitely many closed half-planes. Hence $\mathcal{A}_\text{KS}(K;x)$ and $\mathcal{A}_\text{KS}(K)$  are closed and convex polytopes.

\textbf{Part 4:} We again break this part by first showing that the correspondence $\mathcal{A}_\text{KS}(\cdot;x):[0,1]\rightrightarrows [0,1]^2$ is uhc and lhc.

To show uhc, note that $\mathcal{A}_\text{KS}(\cdot,x)$ is compact-valued since $\mathcal{A}_\text{KS}(K,x)$ is closed and bounded for all $K \in [0,1]$. Let $K_n \to K$, $\textbf{p}_n \in \mathcal{A}_\text{KS}(K_n;x)$ and $\textbf{p}_n \to \textbf{p}$ as $n\to\infty$. We can see that $\textbf{p} \in \mathcal{A}_\text{KS}(K;x)$ because
\begin{align*}
 A_\text{KS}\textbf{p} - a_\text{KS}(K) &= \lim_{n \to \infty} \left(A_\text{KS}\textbf{p}_n - a_\text{KS}(K_n)\right)\leq \lim_{n\to\infty}(0) = 0,
\end{align*}
which follows from the continuity of $A_\text{KS}\textbf{p} - a_\text{KS}(K)$ in $(K,\textbf{p})$. Therefore $\mathcal{A}_\text{KS}(K;x)$ is uhc.

To show lhc, let $K_n \to K \in [0,1]$ and fix $\textbf{p} = (p_0,p_1) \in \mathcal{A}_\text{KS}(K;x)$. $\mathcal{A}_\text{KS}(K;x)$ is lhc if we can find $\textbf{p}_n \in \mathcal{A}_\text{KS}(K_n;x)$ such that $\textbf{p}_n \to \textbf{p}$. 

If $K = 0$, let $\textbf{p}_n = \textbf{p}\in \mathcal{A}_\text{KS}(0;x)$. In this case, $\textbf{p}_n \in \mathcal{A}_\text{KS}(K_n;x)$ since $\mathcal{A}_\text{KS}(0;x) \subseteq \mathcal{A}_\text{KS}(K_n;x)$ for all $n$. Trivially $\textbf{p}_n \to \textbf{p}$ and therefore lhc is established.

When $K> 0$ let
\begin{align*}
 (p_{0n},p_{1n}) = \textbf{p}_n &= \textbf{p} \cdot \min\{K_n/K,1\}.
\end{align*}
First note that $\textbf{p}_n \in [0,1]^2$ since $\textbf{p} \in [0,1]^2$ and $\min\{K_n/K,1\} \in [0,1]$. Also note that $|p_{0n} - p_{1n}| = |p_{0} - p_{1}| \min\{K_n/K,1\} \leq K \cdot \min\{K_n/K,1\} \leq K_n$ so $\textbf{p}_n \in \mathcal{A}_\text{KS}(K_n;x)$. Finally, $\textbf{p}_n \to \textbf{p}$ since $\min\{K_n/K,1\} \to 1$ as $n\to\infty$. Therefore, $\mathcal{A}_\text{KS}(K;x)$ is lhc. We conclude the proof similarly to that of Lemma \ref{lemma:MSM_prim_assn}.
\end{proof}

\begin{proof}[Proof of Theorem \ref{thm:prob_IDset}]
We prove the four claims of the theorem separately.

\textbf{Claim 1:} By Proposition \ref{prop:noassn_discrete}, the identified set for $\textbf{p}_Y$ under Assumption \ref{assn:NonTrivialinstrument} is $\mathcal{H}_0 \times \mathcal{H}_1$. By assumption \ref{assn:ind_relax}, $\textbf{p}_Y$ lies in $\mathcal{A}_0(\theta) \times \mathcal{A}_1(\theta)$. Therefore, the identified set under Assumption \ref{assn:ind_relax} is given by their intersection.

\bigskip

\textbf{Claim 2:} Fix $x \in \{0,1\}$. To show this claim, we first note that the constant correspondence which maps $\theta$ to $\mathcal{H}_x$ is continuous for all $\theta \in [0,1]$, which can be established from the definition of uhc and lhc. Second, we note that $\mathcal{H}_x$ is a closed set. Third, by Exercise 11.18 b in \cite{Border1985}, the set $\Theta_x = \{\theta \in [0,1]: \mathcal{H}_x \cap \mathcal{A}_x(\theta) \neq \emptyset\}$ is closed. By Assumption \ref{assn:ind_relax}.1, $\mathcal{H}_x \cap \mathcal{A}_x(1) = \mathcal{H}_x \neq \emptyset$ so $\Pi_x(1)$ is non-empty. By construction, the set $\Pi_x(\theta)  = \mathcal{H}_x \cap \mathcal{A}_x(\theta)$ is weakly increasing in $\theta$ so the set $\Theta_x$ must be a closed interval of the kind $[\underline{\theta}_x,1]$. The set $\Pi(\theta)$ is non-empty when $\Pi_0(\theta)$ and $\Pi_1(\theta)$ are non-empty, or when $\theta \in [\underline{\theta}_0,1]\cap[\underline{\theta}_1,1]$. This occurs when $\theta \geq \underline{\theta} \coloneqq \max\{\underline{\theta}_0,\underline{\theta}_1\}$.

\bigskip

\textbf{Claim 3:} This follows from $\mathcal{H}_x$ and $\mathcal{A}_x(\theta)$ being closed convex polytopes, and by the fact that polytopes, closedness, and convexity are preserved by finite intersections and Cartesian products.

\bigskip

\textbf{Claim 4:} As shown above, both $\mathcal{A}_x(\theta)$ and $\mathcal{H}_x$ are closed-valued uhc correspondences for $x = 0,1$. By Proposition 11.21.a in \cite{Border1985}, this implies their intersection is a uhc correspondence. By the assumption that $\text{int}(\mathcal{H}_x \cap \mathcal{A}_x(\theta)) \neq \emptyset$ for $\theta > \underline{\theta}$, that both $\mathcal{H}_x$ and $\mathcal{A}_x(\theta)$ are lhc correspondences, and that they are both convex-valued, we can use Theorem B in \cite{LechickiSpakowski1985} to show that $\mathcal{H}_x \cap \mathcal{A}_x(\theta)$ is lhc for $\theta \in (\underline{\theta},1]$. By Theorem 11.25 in \cite{Border1985}, the correspondence $\Pi_0 \times \Pi_1:[\underline{\theta},1] \rightrightarrows [0,1]^4$ is therefore uhc for $\theta \in [\underline{\theta},1]$ and lhc for $\theta \in (\underline{\theta},1]$. 

We finish proving this claim by showing that $\Pi(\theta)$ is also lhc at $\theta = \underline{\theta}$. To see this, let $\theta_n \to \underline{\theta}$ and let $\textbf{p} \in \Pi(\underline{\theta})$. Since $\underline{\theta}$ is the lower bound of the correspondence's domain, we must have that $\theta_n \geq \underline{\theta}$ for all $n$. Let $\textbf{p}_n = \textbf{p}$. By monotonicity of the $\Pi$ correspondence, $\textbf{p}_n = \textbf{p} \in \Pi(\underline{\theta}) \subseteq \Pi(\theta_n)$ for all $n$. Trivially, $\textbf{p}_n \to \textbf{p}$. Therefore, $\Pi$ is uhc and lhc, and hence continuous, for $\theta \in [\underline{\theta},1]$.
\end{proof}

\begin{proof}[Proof of Corollary \ref{corr:hetTrtBinATEbounds}]

\textbf{Claim 1:} By definition, the identified set for $(\Prob(Y(0)=1),\Prob(Y(1)=1))$ is given by
\begin{align*}
 &\{ (p_{00}(1-p_Z) + p_{01}p_Z,p_{10}(1-p_Z) + p_{11}p_Z): \textbf{p} \in \Pi(\theta) = \Pi_0(\theta) \times \Pi_1(\theta)\}\\
 &= \{ (p_{00}(1-p_Z) + p_{01}p_Z,p_{10}(1-p_Z) + p_{11}p_Z): (p_{00},p_{01}) \in \Pi_0(\theta),  (p_{10},p_{11}) \in \Pi_1(\theta)\}\\
 &= \{ p_{00}(1-p_Z) + p_{01}p_Z: (p_{00},p_{01}) \in \Pi_0(\theta)\} \times \{ p_{10}(1-p_Z) + p_{11}p_Z: (p_{10},p_{11}) \in \Pi_1(\theta)\}.
\end{align*}
For $x=0,1$, the set $\Pi_x(\theta)$ is convex and compact, and the function $(p_{x0},p_{x1})\mapsto p_{x0}(1-p_Z) + p_{x1}p_Z$ is continuous. Hence, the function $p_{x0}(1-p_Z) + p_{x1}p_Z$ attains its minimum and maximum, denoted by $\underline{P}_x(\theta)$ and $\overline{P}_x(\theta)$ respectively. By the convexity of $\Pi_x(\theta)$,  all values in $[\underline{P}_x(\theta),\overline{P}_x(\theta)] = I_x(\theta)$ are attained.

\bigskip

\textbf{Claim 2:} By Theorem \ref{thm:prob_IDset}, the correspondence $\Pi(\theta)$ is continuous and compact-valued for $\theta \in [\underline{\theta},1]$. The function $(p_{x0},p_{x1})\mapsto p_{x0}(1-p_Z) + p_{x1}p_Z$ is continuous for $x=0,1$. Therefore, by the Maximum Theorem (\cite{Border1985} Theorem 12.1  or \cite{Berge1959}), $\overline{P}_x(\theta)$ is continuous. Applying this theorem again to the negative of that function yields that $\underline{P}_x(\theta)$ is continuous. Monotonicity of these function follows from $\Pi(\theta) \subseteq \Pi(\theta')$ for $\theta \leq \theta'$.

\bigskip

\textbf{Claim 3:} This follows from the identified set of $(\Prob(Y(0)=1),\Prob(Y(1)=1))$ being a Cartesian product.
\end{proof}

\section{Proofs for Section \ref{sec:cont_Y}: Continuous Outcomes}\label{appsec:cont_outcomes_proofs}
\begin{proof}[Proof of Proposition \ref{prop:cont_noassbounds}]
 We have by the law of total probability that
\begin{align}\label{eq:cont_noassbounds}
	f_{Y}(y\mid z;x) &= f_{Y(x),X|Z}(y,x\mid z) + f_{Y(x),X|Z}(y,1-x\mid z)\notag\\
	&\geq f_{Y(x),X|Z}(y,x\mid z)\notag\\
	&= f_{Y|X,Z}(y\mid x, z)\pi(x \mid z).
\end{align}
These densities are well defined by assumptions \ref{assn:NonTrivialinstrument} and \ref{assn:contsupport}. By Assumption \ref{assn:compactdensities}, $\textbf{f}_{Y}(\cdot \mid z;x) \in \mathcal{F}_{\text{den},x}$. Combining with equation \eqref{eq:cont_noassbounds}, this yields that $\textbf{f}_{Y} \in \mathcal{H}$. To show sharpness, let $\textbf{f} = (f_{00},f_{01},f_{10},f_{11}) \in \mathcal{H}$. For $x=0,1$ define
 \begin{align*}
  f_{Y(0),Y(1)|X,Z}(y_0,y_1\mid x,z) &= f_{Y(0)|X,Z}(y_0\mid x,z)f_{Y(1)|X,Z}(y_1\mid x,z)\\
  f_{Y(0)|X,Z}(y\mid x,z) &= \begin{cases} 
        f_{Y|X,Z}(y\mid 0,z) &\text{ if } x = 0\\
        (f_{0z}(y) - f_{Y,X|Z}(y,0\mid z))/\pi(1\mid z) &\text{ if } x= 1
        \end{cases}\\
  f_{Y(1)|X,Z}(y\mid x,z) &= \begin{cases} 
        f_{Y|X,Z}(y\mid 1,z) &\text{ if } x = 1\\
        (f_{1z}(y) - f_{Y,X|Z}(y,1\mid z))/\pi(0\mid z) &\text{ if } x= 0.      
        \end{cases}
 \end{align*}
By $\textbf{f} \in \mathcal{H}$, these are all non-negative functions that integrate to 1 over $\mathcal{Y}_x$, hence they are probability density functions. They coincide with the observed distributions $f_{Y|X,Z}$ because $f_{Y|X,Z}(y\mid x,z) = f_{Y(x)|X,Z}(y\mid x,z)$.

Also, we have that
\begin{align*}
 f_{Y}(y\mid z;x) &= f_{Y(x),X|Z}(y,x\mid z) + f_{Y(x),X|Z}(y,1-x\mid z)\\
 &= f_{Y(x),X|Z}(y,x\mid z) + f_{Y(x)|X,Z}(y\mid 1-x, z) \pi(1-x \mid z)\\
 &= f_{Y,X|Z}(y,x\mid z) + \left(f_{xz}(y) - f_{Y,X|Z}(y,x\mid z)\right)\\
 &= f_{xz}(y).
\end{align*}
Therefore, this density $f_{Y(0),Y(1)\mid X,Z}$ is consistent with the known conditional distribution $f_{Y|X,Z}$, with $\textbf{f} \in \mathcal{H}$, and with Assumption \ref{assn:contsupport}. 
\end{proof}

\begin{proof}[Proof of Proposition \ref{prop:high-level_relax_cont}]
 This proposition follows from lemmas \ref{lemma:cont_MS_prim_assn}--\ref{lemma:cont_KSdist_prim_assn}, in which we verify that the four conditions of Assumption \ref{assn:ind_relax_cont} hold for the corresponding sensitivity model.
\end{proof}

\begin{lemma}\label{lemma:cont_MS_prim_assn}
Let assumptions \ref{assn:NonTrivialinstrument}, \ref{assn:contsupport}, and \ref{assn:compactdensities} hold. Then, the correspondence defined in equation \eqref{eq:MSM_cont} satisfies Assumption \ref{assn:ind_relax_cont}.
\end{lemma}

\begin{proof}[Proof of Lemma \ref{lemma:cont_MS_prim_assn}]
\textbf{Part 1:} When $\lambda = 0$, we have that
\begin{align*}
 \mathcal{A}_\text{MSM}(0;x) &= \{(f_0,f_1) \in \mathcal{F}_{\text{den},x}^2: -f_{0} + f_{1} \leq 0, f_{0} - f_{1} \leq 0\} = \{\textbf{f} \in \mathcal{F}_{\text{den},x}^2: f_{0} = f_{1}\}
\end{align*}
and that
\begin{align*}
 \mathcal{A}_\text{MSM}(1;x) &= \{(f_0,f_1) \in \mathcal{F}_{\text{den},x}^2: -f_{0x} \leq 0,  - f_{1x}\leq 0\} = \mathcal{F}_{\text{den},x}^2.
\end{align*}

\textbf{Part 2:} Suppose $(f_0,f_1) \in \mathcal{A}_\text{MSM}(\lambda;x)$ and let $\lambda' \in [\lambda,1]$. Then, since densities are non-negative,
\begin{align*}
 A_\text{MSM}(\lambda')\textbf{f} = \begin{pmatrix}
 - f_0 + f_1 - \lambda' f_1\\
 - f_1 + f_0 - \lambda' f_0
\end{pmatrix} \leq \begin{pmatrix}
 - f_0 + f_1 - \lambda f_1\\
 - f_1 + f_0 - \lambda f_0
\end{pmatrix} \leq \begin{pmatrix}
 0\\
 0
 \end{pmatrix}.
\end{align*}
Therefore, $\textbf{f} \in \mathcal{A}_\text{MSM}(\lambda';x)$.

\textbf{Part 3:} To show $\mathcal{A}_\text{MSM}(\lambda;x)$ is closed, suppose that $\textbf{f}_n = (f_{0n},f_{1n}) \in \mathcal{A}_\text{MSM}(\lambda;x)$ and $\textbf{f}_n \to \textbf{f} = (f_0,f_1) \in \mathcal{F}_{\text{den},x}^2$ in sup norm as $n\rightarrow\infty$. We show that $\textbf{f} \in \mathcal{A}_\text{MSM}(\lambda;x)$. To see this, note that sup norm convergence implies pointwise convergence, and therefore
\begin{align*}
 \begin{pmatrix}
 - f_0(y) + (1-\lambda)f_1(y)\\
 - f_1(y) + (1-\lambda)f_0(y)
\end{pmatrix} &= \lim_{n \to \infty} \begin{pmatrix}
 - f_{0n}(y) + (1-\lambda)f_{1n}(y)\\
 - f_{1n}(y) + (1-\lambda)f_{0n}(y)
\end{pmatrix} \leq \begin{pmatrix}
0\\
0
\end{pmatrix}
\end{align*}
for all $y \in \mathcal{Y}_x$.

It is characterized by finitely many weak componentwise inequalities by construction, and it is convex because $\mathcal{F}_{\text{den},x}$ is convex (Assumption \ref{assn:compactdensities}) and by the fact that it is characterized by finitely many linear inequalities.

\textbf{Part 4:} We break this part in two and first show the correspondence is uhc followed by lhc.

To show uhc, let $\lambda_n \to \lambda$, $\textbf{f}_n \in \mathcal{A}_\text{MSM}(\lambda_n;x)$, and $\textbf{f}_n \to \textbf{f}$ in sup-norm. The correspondence is uhc at $\lambda$ if $\textbf{f} \in \mathcal{A}_\text{MSM}(\lambda;x)$. This is the case because
\begin{align*}
 \begin{pmatrix}
 - f_0(y) + (1-\lambda)f_1(y)\\
 - f_1(y) + (1-\lambda)f_0(y)
\end{pmatrix} &= \lim_{n \to \infty} \begin{pmatrix}
 - f_{0n}(y) + (1-\lambda_n)f_{1n}(y)\\
 - f_{1n}(y) + (1-\lambda_n)f_{0n}(y)
\end{pmatrix} \leq 0
\end{align*}
for all $y \in \mathcal{Y}_x$, where the equality follows from the pointwise convergence of $(\textbf{f}_n(y),\lambda_n)$ to $(\textbf{f}(y),\lambda)$.

To show lhc, let $\lambda_n \to \lambda$ and $\textbf{f} = (f_0,f_1) \in \mathcal{A}_\text{MSM}(\lambda;x)$. We aim to find $\textbf{f}_n = (f_{0n},f_{1n}) \in \mathcal{A}_\text{MSM}(\lambda_n;x)$ such that $\|\textbf{f}_n - \textbf{f}\|_\infty \to 0$. If $\lambda = 0$, then let $\textbf{f}_n = \textbf{f}$, where $\textbf{f} \in \mathcal{A}_\text{MSM}(0;x) \subseteq \mathcal{A}_\text{MSM}(\lambda_n,x)$ for all $\lambda_n$. Therefore, $\mathcal{A}_\text{MSM}(\lambda;x)$ is lhc at $\lambda = 0$.

Let $\lambda \in (0,1]$ and 
\begin{align*}
 \varepsilon_n &= \max\left\{1 - \frac{\lambda_n}{\lambda},\frac{\lambda - \lambda_n}{\lambda(1-\lambda_n)},0\right\}\\
 f_{0n} &= (1-\varepsilon_n)f_0 + \varepsilon_n f_1\\
 f_{1n} &= f_1.
\end{align*}
We see that $\varepsilon_n \to 0$ as $\lambda_n \to \lambda$. Trivially, $\varepsilon_n \geq 0$. $\varepsilon_n \leq 1$ because $1 - \lambda_n/\lambda \leq 1$ and because $\frac{\lambda - \lambda_n}{\lambda(1-\lambda_n)} \leq \frac{\lambda - \lambda\lambda_n}{\lambda(1-\lambda_n)} = 1$. Therefore, $(1-\varepsilon_n)f_0 + \varepsilon_n f_1$ is a convex combination of $f_0$ and $f_1$ which implies $f_{0n} \in \mathcal{F}_{\text{den},x}$ by the convexity of $\mathcal{F}_{\text{den},x}$.  

The first inequality characterizing the Marginal Sensitivity Model is satisfied at $(\lambda_n,\textbf{f}_n)$ because
\begin{align*}
 - f_{0n} + (1-\lambda_n)f_{1n} &= - (1-\varepsilon_n)f_0 - \varepsilon_n f_1 + (1-\lambda_n)f_1\\
 &= (1-\varepsilon_n)\left(-f_0 + \frac{1-\lambda_n - \varepsilon_n}{1-\varepsilon_n}f_1\right)\\
 &= \underbrace{(1-\varepsilon_n)}_{\geq 0}\underbrace{\left(-f_0 + (1-\lambda) f_1\right)}_{\leq 0} + f_1 \left( -\lambda_n  + \lambda  - \lambda\varepsilon_n \right)\\
 &\leq f_1 \left( -\lambda_n  + \lambda  - \lambda\max\left\{1 - \frac{\lambda_n}{\lambda},\frac{\lambda - \lambda_n}{\lambda(1-\lambda_n)},0\right\} \right)\\
 &\leq  \underbrace{f_1}_{\geq 0}\underbrace{ \left( -\lambda_n  + \lambda  - \lambda\left(1 - \frac{\lambda_n}{\lambda}\right) \right)}_{=0}\\
 &= 0.
\end{align*}
The first inequality follows from $\varepsilon_n \leq 1$ and $\textbf{f}\in \mathcal{A}_\text{MSM}(\lambda;x)$. The second follows from $f_1 \geq 0$ and from the definition of $\varepsilon_n$. Therefore, $\textbf{f}_n$ satisfies the first inequality.

It also satisfies the second inequality because
\begin{align*}
 (1-\lambda_n)f_{0n} - f_{1n} &= (1-\lambda_n)(1-\varepsilon_n)f_0 - (1-(1-\lambda_n)\varepsilon_n)f_1\\
 &= (1-(1-\lambda_n)\varepsilon_n)\left(\left(\frac{(1-\lambda_n)(1-\varepsilon_n)}{1 - (1-\lambda_n)\varepsilon_n} - (1-\lambda)\right)f_0 + ((1-\lambda)f_0 - f_1) \right)\\
 &= (1-(1-\lambda_n)\varepsilon_n)((1-\lambda)f_0 - f_1) + f_0((1-\lambda_n)(1-\varepsilon_n) - (1-\lambda)(1-(1-\lambda_n)\varepsilon_n))\\
 &\leq f_0\left(-\max\left\{1 - \frac{\lambda_n}{\lambda},\frac{\lambda - \lambda_n}{\lambda(1-\lambda_n)},0\right\}\lambda(1-\lambda_n) + \lambda - \lambda_n\right)\\
 &\leq f_0\left(-\frac{\lambda - \lambda_n}{\lambda(1-\lambda_n)}\lambda(1-\lambda_n) + \lambda - \lambda_n\right)\\
 &= 0.
\end{align*}
The first inequality follows from $\varepsilon_n \leq 1$ and $\textbf{f}\in \mathcal{A}_\text{MSM}(\lambda;x)$. The second follows from $f_1 \geq 0$ and from the definition of $\varepsilon_n$. Therefore, $\textbf{f}_n$ satisfies both inequalities. This implies that $\mathcal{A}_\text{MSM}(\lambda;x)$ is lhc and concludes the proof of Part 4.
\end{proof}

\begin{lemma}\label{lemma:cont_cdep_prim_assn}
Let assumptions \ref{assn:NonTrivialinstrument}, \ref{assn:contsupport}, and \ref{assn:compactdensities} hold. Then, the correspondence defined in equation \eqref{eq:cdep_cont} satisfies Assumption \ref{assn:ind_relax_cont}.
\end{lemma}

\begin{proof}[Proof of Lemma \ref{lemma:cont_cdep_prim_assn}]
\textbf{Part 1:} When $c = 0$, we have that
\begin{align*}
 \mathcal{A}_\text{$c$-dep}(0;x) &= \{(f_0,f_1) \in \mathcal{F}_{\text{den},x}^2: -f_{0} + f_{1} \leq 0, f_{0} - f_{1} \leq 0 \} = \{\textbf{f} \in \mathcal{F}_{\text{den},x}^2: f_{0} = f_{1}\}
\end{align*}
and that
\begin{align*}
 \mathcal{A}_\text{$c$-dep}(1;x) &= \{(f_0,f_1) \in \mathcal{F}_{\text{den},x}^2: -f_{0} \leq 0,  - f_{1} \leq 0 \} = \mathcal{F}_{\text{den},x}^2.
\end{align*}

\textbf{Part 2:} Suppose $\textbf{f} = (f_0,f_1) \in \mathcal{A}_\text{$c$-dep}(c;x)$ and let $c' \in [c,1]$. Then, since $k_z(c)$ is nonincreasing,
\begin{align*}
 A_\text{$c$-dep}(c')\textbf{f} &= \begin{pmatrix}
 - f_0 + k_1(c')f_1 \\
 - f_1 + k_0(c')f_0
\end{pmatrix} \leq \begin{pmatrix}
 - f_0 + k_1(c)f_1 \\
 - f_1 + k_0(c)f_0
\end{pmatrix} \leq \begin{pmatrix}
 0\\
 0
 \end{pmatrix}.
\end{align*}
Therefore, $\textbf{f} \in \mathcal{A}_\text{$c$-dep}(c';x)$.

\textbf{Part 3:} We show $\mathcal{A}_\text{$c$-dep}(c;x)$ is closed following the same arguments as in the proof of Lemma \ref{lemma:cont_MS_prim_assn} and the continuity of $k_z(c)$ in $c$.

\textbf{Part 4:} We break this part into two and first show the correspondence is uhc followed by lhc.

To show uhc, let $c_n \to c$, $\textbf{f}_n \in \mathcal{A}_\text{$c$-dep}(c_n;x)$, and $\textbf{f}_n \to \textbf{f}$ in sup-norm. The correspondence is uhc at $c$ if $\textbf{f} \in \mathcal{A}_\text{$c$-dep}(c;x)$. This is the case because
\begin{align*}
 \begin{pmatrix}
 - f_0(y) + k_1(c)f_1(y)\\
 - f_1(y) + k_0(c)f_0(y)
\end{pmatrix} &= \lim_{n \to \infty} \begin{pmatrix}
 - f_{0n}(y) + k_1(c_n)f_{1n}(y)\\
 - f_{1n}(y) + k_0(c_n)f_{0n}(y)
\end{pmatrix} \leq \begin{pmatrix}
0\\
0
\end{pmatrix}
\end{align*}
where the equality follows from the point-wise convergence of $(\textbf{f}_n(y),c_n)$ to $(\textbf{f}(y),c)$ and the continuity of $k_0$ and $k_1$.

To show lhc, let $c_n \to c$ and $\textbf{f}\in \mathcal{A}_\text{$c$-dep}(c;x)$. We aim to find $\textbf{f}_n \in \mathcal{A}_\text{$c$-dep}(c_n;x)$ such that $\textbf{f}_n \to \textbf{f}$. $\mathcal{A}_\text{$c$-dep}(c;x)$ is lhc at $c=0$ following the same arguments as in the proof of Lemma \ref{lemma:cont_MS_prim_assn}.

Let $c \in (0,1]$ and 
\begin{align*}
 \varepsilon_n &= \max\left\{1 - \frac{1-k_1(c_n)}{1-k_1(c)},\frac{k_0(c_n) - k_0(c)}{(1-k_0(c))k_0(c_n)},0\right\}\\
 f_{0n} &= (1-\varepsilon_n)f_0 + \varepsilon_n f_1\\
 f_{1n} &= f_1.
\end{align*}
By the continuity of $k_0$ and $k_1$, we see that $\varepsilon_n \to 0$ as $c_n \to c$. Trivially, $\varepsilon_n \geq 0$. $\varepsilon_n \leq 1$ because $1 - \frac{1-k_1(c_n)}{1-k_1(c)} \leq 1$ and because $\frac{k_0(c_n) - k_0(c)}{(1-k_0(c))k_0(c_n)} \leq \frac{k_0(c_n) - k_0(c)k_0(c_n)}{(1-k_0(c))k_0(c_n)} = 1$. Therefore, $(1-\varepsilon_n)f_0 + \varepsilon_n f_1 = f_{0n}$ is a convex combination of elements of $\mathcal{F}_{\text{den},x}$, hence $f_{0n} \in \mathcal{F}_{\text{den},x}$. 

The first inequality characterizing $c$-dependence is satisfied at $(c_n,\textbf{f}_n)$ because
\begin{align*}
 - f_{0n} + k_1(c_n)f_{1n} &= - (1-\varepsilon_n)f_0 - \varepsilon_n f_1 + k_1(c_n)f_1\\
 &= (1-\varepsilon_n)\left(-f_0 + \frac{k_1(c_n) - \varepsilon_n}{1-\varepsilon_n}f_1\right)\\
 &= \underbrace{(1-\varepsilon_n)}_{\geq 0}\underbrace{\left(-f_0 + k_1(c) f_1\right)}_{\leq 0} + f_1 \left( k_1(c_n) - k_1(c)  - (1-k_1(c_n))\varepsilon_n \right)\\
 &\leq f_1 \left( k_1(c_n) - k_1(c)  - (1-k_1(c_n))\max\left\{1 - \frac{1-k_1(c_n)}{1-k_1(c)},\frac{k_0(c_n) - k_0(c)}{(1-k_0(c))k_0(c_n)},0\right\} \right)\\
 &\leq  \underbrace{f_1}_{\geq 0}\underbrace{ \left( k_1(c_n) - k_1(c)  - (1-k_1(c_n))\left(1 - \frac{1-k_1(c_n)}{1-k_1(c)}\right) \right)}_{=0}\\
 &= 0.
\end{align*}
The first inequality follows from $\varepsilon_n \leq 1$ and $\textbf{f}\in \mathcal{A}_\text{$c$-dep}(c;x)$. The second follows from $f_1 \geq 0$ and from the definition of $\varepsilon_n$. Therefore, $\textbf{f}_n$ satisfies the first inequality.

It also satisfies the second inequality because
\begin{align*}
 k_0(c_n)f_{0n} - f_{1n} &= k_0(c_n)(1-\varepsilon_n)f_0 - (1-k_0(c_n)\varepsilon_n)f_1\\
 &= (1-k_0(c_n)\varepsilon_n)\left(\left(\frac{k_0(c_n)(1-\varepsilon_n)}{1 - k_0(c_n)\varepsilon_n} - k_0(c)\right)f_0 + (k_0(c)f_0 - f_1) \right)\\
 &= (1-k_0(c_n)\varepsilon_n)(k_0(c)f_0 - f_1) + f_0(k_0(c_n)(1-\varepsilon_n) - k_0(c)(1-k_0(c_n)\varepsilon_n))\\
 &\leq f_0\left(-\max\left\{1 - \frac{1-k_1(c_n)}{1-k_1(c)},\frac{k_0(c_n) - k_0(c)}{(1-k_0(c))k_0(c_n)},0\right\}(1-k_0(c))k_0(c_n) + k_0(c_n) - k_0(c)\right)\\
 &\leq f_0\left(-\frac{k_0(c_n) - k_0(c)}{(1-k_0(c))k_0(c_n)}(1-k_0(c))k_0(c_n) + k_0(c_n) - k_0(c)\right)\\
 &= 0.
\end{align*}
The first inequality follows from $\varepsilon_n \leq 1$ and $\textbf{f}\in \mathcal{A}_\text{$c$-dep}(c;x)$. The second follows from $f_1 \geq 0$ and from the definition of $\varepsilon_n$. Therefore, $\textbf{f}_n$ satisfies both inequalities, which implies that $\mathcal{A}_\text{$c$-dep}(c;x)$ is lhc. This concludes the proof of Part 4.
\end{proof}

\begin{lemma}\label{lemma:cont_KSdist_prim_assn}
Let assumptions \ref{assn:NonTrivialinstrument}, \ref{assn:contsupport}, and \ref{assn:compactdensities} hold. Then, the correspondence defined in equation \eqref{eq:KS_cont} satisfies Assumption \ref{assn:ind_relax_cont}.
\end{lemma}

\begin{proof}[Proof of Lemma \ref{lemma:cont_KSdist_prim_assn}]
\textbf{Part 1:} When $K = 0$, we have that
\begin{align*}
 \mathcal{A}_\text{KS}(0;x) &= \{(f_0,f_1) \in \mathcal{F}_{\text{den},x}^2: f_{0} - f_{1} \leq 0, -f_{0} + f_{1} \leq 0  \} = \{\textbf{f} \in \mathcal{F}_{\text{den},x}^2: f_{0} = f_{1}\}
\end{align*}
and that
\begin{align*}
 \mathcal{A}_\text{KS}(1;x) &= \{(f_0,f_1) \in \mathcal{F}_{\text{den},x}^2: f_{0} - f_{1} \leq +\infty,  -f_{0} + f_{1} \leq +\infty \} = \mathcal{F}_{\text{den},x}^2
\end{align*}
since the densities are bounded by Assumption \ref{assn:compactdensities}.

\textbf{Part 2:} Suppose $(f_0,f_1) \in \mathcal{A}_\text{KS}(K;x)$ and let $K' \in [K,1]$. Then, 
\begin{align*}
 |f_{0} - f_{1}| &\leq \frac{K}{1-K} \leq \frac{K'}{1-K'}
\end{align*}
Therefore, $\textbf{f} \in \mathcal{A}_\text{KS}(K';x)$.

\textbf{Part 3:} To show $\mathcal{A}_\text{KS}(K;x)$ is closed, let $\textbf{f}_n = (f_{0n},f_{1n})\in \mathcal{A}_\text{KS}(K;x)$ converge in the sup norm to $\textbf{f} = (f_0,f_1) \in \mathcal{F}_{\text{den},x}^2$. We show that $\textbf{f} \in \mathcal{A}_\text{KS}(K;x)$. By uniform convergence,
\begin{align*}
 |f_{0}(y) - f_{1}(y)| &= \lim_{n \to \infty} |f_{0n}(y) - f_{1n}(y)| \leq \frac{K}{1-K}
\end{align*}
so $\textbf{f} \in \mathcal{A}_\text{KS}(K;x)$. It is convex because it is characterized by finitely many componentwise weak inequalities.

\textbf{Part 4:} We again break this part into two and first show the correspondence is uhc followed by lhc.

To show uhc, let $K_n \to K \in [0,+\infty]$, $\textbf{f}_n \in \mathcal{A}_\text{KS}(K_n;x)$, and $\textbf{f}_n \to \textbf{f}$ in sup-norm. The correspondence is uhc at $K$ if $\textbf{f} \in \mathcal{A}_\text{KS}(K;x)$. This is the case because
\begin{align*}
 |f_{0}(y) - f_{1}(y)| - \frac{K}{1-K} &= \lim_{n \to \infty} \left(|f_{0n}(y) - f_{1n}(y)| - \frac{K_n}{1-K_n}\right) \leq 0
\end{align*}
where the equality follows from the pointwise convergence of $(\textbf{f}_n(y),K_n)$ to $(\textbf{f}(y),K)$ for all $y \in \mathcal{Y}_x$. 

To show lhc, let $K_n \to K$ and $\textbf{f} = (f_0,f_1) \in \mathcal{A}_\text{KS}(K;x)$. We aim to find $\textbf{f}_n = (f_{0n},f_{1n}) \in \mathcal{A}_\text{KS}(K_n;x)$ such that $\textbf{f}_n \to \textbf{f}$. If $K = 0$, then let $\textbf{f}_n = \textbf{f}$, where $\textbf{f} \in \mathcal{A}_\text{KS}(0;x) \subseteq \mathcal{A}_\text{KS}(K_n,x)$ for all $K_n$. Therefore, $\mathcal{A}_\text{KS}(K;x)$ is lhc at $K = 0$.

Let $K \in (0,1)$ and 
\begin{align*}
 \varepsilon_n &= \min\left\{\frac{K_n/(1-K_n)}{K/(1-K)},0\right\}\\
 f_{0n} &= \varepsilon_n f_0 + (1-\varepsilon_n) f_1\\
 f_{1n} &= f_1.
\end{align*}
We see that $\varepsilon_n \in [0,1]$ and $\varepsilon_n \to 0$ as $K_n \to K$. Therefore, $\varepsilon_n f_0 + (1-\varepsilon_n) f_1 \in \mathcal{F}_{\text{den},x}$ because $\mathcal{F}_{\text{den},x}$ is convex. 

We have that
\begin{align*}
 \|f_{0n} - f_{1n}\|_\infty &= \varepsilon_n\|f_{0} - f_{1}\|_\infty \leq \frac{K_n/(1-K_n)}{K/(1-K)}\|f_{0} - f_{1}\|_\infty \leq K_n/(1-K_n),
\end{align*}
so $\textbf{f}_n \in \mathcal{A}_\text{KS}(K_n;x)$. We also have that $\|f_{0n} - f_0\|_\infty = \varepsilon_n \|f_0 - f_1\|_\infty \to 0$. 

The case where $K  = 1$ can also be shown by letting $\varepsilon_n = K_n/((1-K_n)\|f_1 - f_0\|_\infty)\mathbbm{1}(\|f_1 - f_0\|_\infty \neq 0)$ and recalling that $\|f_1 - f_0\|_\infty \leq \|f_1\|_\infty + \|f_0\|_\infty < \infty$ by Assumption \ref{assn:compactdensities}.

Therefore, $\mathcal{A}_\text{KS}(K;x)$ is lhc at $K$. This concludes the proof of Part 4.
\end{proof}

\begin{lemma}\label{lemma:compact_subset}
 Let Assumption \ref{assn:compactdensities} hold. Then $\mathcal{F}_{\text{den},x}$ is compact under $\|\cdot\|_\infty$.
\end{lemma}

\begin{proof}[Proof of Lemma \ref{lemma:compact_subset}]
 $\mathcal{F}_{\text{den},x}$ is a subset of compact set $\mathcal{F}_x(\mathcal{Y}_x)$ and therefore it is relatively compact. To show its compactness, we show that $\mathcal{F}_{\text{den},x}$ is closed. To show this, let $f_n \in \mathcal{F}_{\text{den},x}$ such that $\|f_n - f\|_\infty \to 0$ for some $f \in \mathcal{F}_x(\mathcal{Y}_x)$ as $n\to\infty$. We show $f \in \mathcal{F}_{\text{den},x}$, hence $\mathcal{F}_{\text{den},x}$ is closed, and thus compact.

To see this is the case, note that
\begin{align*}
 \left|\int_{\mathcal{Y}_x} f(y)dy - 1\right| &= \left|\int_{\mathcal{Y}_x} f(y)dy - \int_{\mathcal{Y}_x} f_n(y)dy\right|
 \leq \int_{\mathcal{Y}_x} |f(y) - f_n(y)|dy\leq \|f - f_n\|_\infty \left(\int_{\mathcal{Y}_x} 1\cdot dy\right).
\end{align*}
Since $\|f_n - f\|_\infty \to 0$ and $\mathcal{Y}_x$ is bounded, the right-hand side can be made arbitrarily small, and we have that $\int_{\mathcal{Y}_x} f(y)dy = 1$.

Also, by uniform convergence we have that $f_n(y) \to f(y)$ for all $y \in \mathcal{Y}_x$. Since $f_n(y) \geq 0$ for all $n$, we also have that $f(y) \geq 0$. Therefore, $f \in \mathcal{F}_{\text{den},x}$ and the proof is complete.
\end{proof}

\begin{proof}[Proof of Theorem \ref{thm:cont_prob_IDset}]
We prove the three claims of the theorem separately.

\textbf{Claim 1:} By Proposition \ref{prop:cont_noassbounds}, the identified set for $\textbf{f}_{Y}$ under assumptions \ref{assn:NonTrivialinstrument}, \ref{assn:contsupport}, and \ref{assn:compactdensities} is $\mathcal{H}_0 \times \mathcal{H}_1$. By Assumption \ref{assn:ind_relax_cont}, $\textbf{f}_{Y}$ lies in $\mathcal{A}_0(\theta) \times \mathcal{A}_1(\theta)$. Therefore, the identified set under assumptions \ref{assn:NonTrivialinstrument}, \ref{assn:contsupport}, \ref{assn:compactdensities}, and \ref{assn:ind_relax_cont} is given by their intersection.

\bigskip

\textbf{Claim 2:} To show this claim, we first note that the constant correspondence which maps $\theta$ to $\mathcal{H}_x$ is continuous for all $\theta \in [0,1]$, which can be directly established from the definition of uhc and lhc. Second, we note that $\mathcal{H}_x$ and $\mathcal{A}_x(\theta)$ are closed sets under $\|\cdot\|_\infty$. The sets $\mathcal{A}_x(\theta)$ and $\mathcal{H}_x$ are compact because they are closed subsets of $\mathcal{F}^2_{\text{den},x}$, which is compact by Lemma \ref{lemma:compact_subset}. Therefore, $\mathcal{A}_x(\theta)$ is compact-valued.  By Theorem 17.25.2 in \cite{AliprantisBorder2006}, $\Pi_x(\theta)$ is uhc.

By the theorem assumption that $\mathcal{F}_{\text{den},x}^2 \cap \mathcal{H}_x \neq \emptyset$, we have that $\Pi_x(1) \neq \emptyset$. By the monotonicity of $\Pi_x(\theta)$ in $\theta$, there exists $\underline{\theta}_x$ such that $\Pi_x(\theta) = \emptyset$ for $\theta < \underline{\theta}_x$ and $\Pi_x(\theta) \neq \emptyset$ for $\theta > \underline{\theta}_x$.  By $\Pi_x(1) \neq \emptyset$, $\underline{\theta}_x \in [0,1]$. Let $\theta_n$ be a nonincreasing sequence in $[0,1]$ converging to $\underline{\theta}_x$. By the sequential definition of uhc (Theorem 17.20 in \cite{AliprantisBorder2006}), the sequence $\textbf{f}_{n} \in \Pi_x(\theta_n)$ has a limit point $\textbf{f}$ in $\Pi_x(\underline{\theta}_x)$. By compactness of $\mathcal{F}_{\text{den},x}^2$, we can extract a subsequence $\textbf{f}_{n_k}$ converging to $\textbf{f}$. Since $\Pi_x$ has closed graphs, we conclude that  $\textbf{f} \in \Pi_x(\underline\theta_x)$, so it is non-empty. This implies $\Pi_x(\theta)$ is non-empty if and only if $\theta \in [\underline{\theta}_x,1]$. The set $\Pi(\theta)$ is non-empty when $\Pi_0(\theta)$ and $\Pi_1(\theta)$ are both non-empty, which occurs when $\theta \in [\underline{\theta}_0,1]\cap[\underline{\theta}_1,1]$ or when $\theta \geq \underline{\theta} \coloneqq \max_{x =0,1}\underline{\theta}_x$.

\bigskip

\textbf{Claim 3:} As shown above, $\Pi_x(\theta)$ are compact-valued uhc correspondences for $x = 0,1$. By the assumption that $\text{int}(\mathcal{H}_x \cap \mathcal{A}_x(\theta)) \neq \emptyset$ for $\theta > \underline{\theta}$, that both $\mathcal{H}_x$ and $\mathcal{A}_x(\theta)$ are lhc correspondences, and that they are both convex-valued, we can use Theorem B in \cite{LechickiSpakowski1985} to show that $\mathcal{H}_x \cap \mathcal{A}_x(\theta)$ is lhc for $\theta \in (\underline{\theta},1]$. By Theorem 17.28 in \cite{AliprantisBorder2006}, this implies their product is a uhc correspondence for $\theta \in [\underline{\theta},1]$ and lhc for $\theta \in (\underline{\theta},1]$. 

We finish this proof by claiming $\Pi(\theta)$ is also lhc at $\theta = \underline{\theta}$. This can be established in the same manner as in the proof of Theorem \ref{thm:prob_IDset}
\end{proof}

\begin{proof}[Proof of Corollary \ref{corr:cont_functionalbounds}]

\textbf{Claim 1:} This follows from $f_{Y(x)}(y) = f_{Y(x)|Z}(y \mid Z=0)(1-p_Z) + f_{Y(x)|Z}(y\mid 1)p_Z$ and the fact that $\Pi(\theta)$ is a Cartesian product.

\bigskip

\textbf{Claim 2:} We have that
\begin{align*}
 \sup_{\textbf{f} \in \Pi(\theta)} \Gamma(\textbf{f}) &= \sup_{\textbf{f}_0 \in \Pi_0(\theta), \textbf{f}_1 \in \Pi_1(\theta)} \left(\int_{\mathcal{Y}_0} \omega_0(y)'\textbf{f}_0(y)dy + \int_{\mathcal{Y}_1} \omega_1(y)'\textbf{f}_1(y)dy\right)\\
 &= \sup_{\textbf{f}_0 \in \Pi_0(\theta)} \int_{\mathcal{Y}_0} \omega_0(y)'\textbf{f}_0(y)dy +  \sup_{\textbf{f}_1 \in \Pi_1(\theta)}\int_{\mathcal{Y}_1} \omega_1(y)'\textbf{f}_1(y)dy = \overline{\Gamma}(\theta).
\end{align*}
A similar argument yields the expression for $\underline{\Gamma}(\theta)$. Therefore, $\Gamma(\textbf{f}_{Y_\cdot|Z}) \in [\underline{\Gamma}(\theta), \overline{\Gamma}(\theta)]$. We now show that this interval is sharp. The endpoints can be attained because they correspond to the maximum/minimum of the continuous function $\Gamma(\cdot)$ over a compact domain $\Pi(\theta)$, and by the extreme value theorem. The interior of this interval can be attained by the convexity of the constraint set $\Pi(\theta)$, which follows from the convexity of $\mathcal{F}_{\text{den},x}$, $\mathcal{A}_x(\theta)$, and $\mathcal{H}_x$.

\bigskip

\textbf{Claim 3:} By Theorem \ref{thm:cont_prob_IDset}.3, the correspondence $\Pi$ is continuous on $[\underline{\theta},1]$. Its values are compact because both $\mathcal{H}_x$ and $\mathcal{A}_x(\theta)$ are compact-valued by assumptions \ref{assn:compactdensities} and derivations in the proof of Theorem \ref{thm:cont_prob_IDset}. $\Pi(\theta)$ is also non-empty for $\theta \in [\underline{\theta},1]$ by construction. By Theorem 17.31 in \cite{AliprantisBorder2006}, the Maximum Theorem for infinite-dimensional spaces, the functions $\overline{\Gamma}(\theta)$ and $\underline{\Gamma}(\theta)$ are continuous.

These functions are monotonic by the sets $\mathcal{A}_x(\theta)$ being monotonic in $\theta$: see Assumption \ref{assn:ind_relax_cont}.2.
\end{proof}

\begin{proof}[Proof of Proposition \ref{prop:approxIDset}]
Recall that $\mathcal{F}_M^{s_Z} = \{W\mathbf{b}^M : W \in \bar{\Delta}_M^{s_Z}\}$, and arranging the $N$ constraints into a matrix with $B^{M,N} = \begin{bmatrix} \mathbf{b}^M(y_1) & \cdots & \mathbf{b}^M(y_{N}) \end{bmatrix}$, we can rewrite the constraint in the definition of $\mathcal{A}^{M,N}(\theta)$ as,
\begin{align*}
\mathcal{A}^{M,N}(\theta) 
&= \left\{ W \mathbf{b}^M : W \in \bar{\Delta}_{M}^{s_Z}, \ A(\theta) W B^{M, N} \le a(\theta)\iota_N^{\top}  \right\} \\
&= \left\{ W \mathbf{b}^M : W \in \bar{\Delta}_{M}^{s_Z}, \ \left( (B^{M,N})^{\top} \otimes A(\theta) \right) \operatorname{vec}(W) \le \iota_N \otimes a(\theta) \right\}
\end{align*}
Turning next to $\mathcal{H}_x^M$, recall that $\mathcal{H}_x^M = \{ \textbf{f} = (f_1,\ldots,f_{s_Z}) \in \mathcal{F}_M^{s_Z} : f_j(y) \geq \pi(x \mid z_j) (B_M f_{Y \mid X, Z})( y \mid x, z_j) \text{ for } j = 1, \ldots, s_Z \text{ and } y \in [0,1]\}$.  We rewrite the inequality constraint $f_{Y(x)|Z}(y \mid z) \ge f_{Y|X,Z}(y \mid x , z)\pi(x \mid z)$ as an equality constraint. To do so, first note that
\begin{align*}
 f_{Y(x) \mid Z}(\cdot \mid z) 
 &= f_{Y, X \mid Z}(\cdot, x \mid z) + f_{Y(x), X \mid Z}(\cdot, 1 - x \mid z) \\
 &= f_{Y \mid X, Z}(\cdot \mid x, z) \pi(x \mid z) + f_{Y(x) \mid X, Z}(\cdot \mid 1 - x, z) \pi(1 - x \mid z).
\end{align*}
The inequality followed from $f_{Y(x) \mid X, Z}(\cdot \mid 1 - x, z) \pi(1 - x \mid z) \geq 0$. Alternatively, we can represent this as 
$$
 \mathbf{f} = \textbf{D}_x \mathbf{f}_{Y}(\cdot ; x) + \textbf{D}_{1-x} \mathbf{q}
$$
for some $\mathbf{q} \in \mathcal{F}_{\text{den}}^{s_Z}$ where $\mathbf{f}_{Y}(\cdot ; x) \coloneqq (f_{Y \mid X, Z}(\cdot \mid x, z_1), \ldots, f_{Y \mid X, Z}(\cdot \mid x, z_{s_Z}))$.

In the approximate constraint set, we replace $f_{Y \mid X, Z}(\cdot \mid x, z)$ by $(B_M f_{Y \mid X, Z})(\cdot \mid x, z)$, and impose that $\mathbf{q} \in \mathcal{F}_M^{s_Z}$. Using standard results on Bernstein polynomials, we have that,
$$
 (B_M f_{Y \mid X, Z})(\cdot \mid x, z) = \sum_{m=0}^M f_{Y \mid X, Z}\left( \frac{m}{M} \mid x, z\right) b_m^M(\cdot).
$$
Hence, gathering the terms $f_{Y \mid X, Z}\left( \frac{m}{M} \mid x, z_j\right)$ into the $s_Z \times (M+1)$ matrix $\Xi^M_x$, we can rewrite the constraint as,
$$
 \mathcal{H}_x^M = \{ \textbf{D}_x \Xi^M_x + \textbf{D}_{1-x} W_{x, 1-x} : W_{x, 1-x} \in \bar{\Delta}_{M}^{s_Z} \}.
$$
\end{proof}
\end{document}